\documentclass[12pt,reqno]{amsart}
\usepackage{mathrsfs}
\usepackage{amssymb,amsthm,amsmath}
\usepackage[numbers,sort&compress]{natbib}
\usepackage{amssymb,amsmath}
\usepackage{amsfonts}
\usepackage{mathrsfs}
\usepackage{latexsym}
\usepackage{amssymb}
\usepackage{amsthm}
\usepackage{color}
\usepackage{pdfsync}
\usepackage{indentfirst}
\usepackage{graphics}
\usepackage{subfigure}
\usepackage{epsfig}
\usepackage{float}
\date{today}

\usepackage{hyperref}
\hypersetup{hypertex=true,
	colorlinks=true,
	linkcolor=blue,
	anchorcolor=g,
	citecolor=red}

\usepackage{amsmath}
\usepackage{amsthm}
   
   \newtheorem{remark}{Remark}

   		\newtheorem{theorem}{Theorem}
      		\newtheorem{proposition}{Proposition}

\newcommand{\R}{\mathbb R}

\makeatletter

\newcommand{\Rmnum}[1]{\expandafter\@slowromancap\romannumeral #1@}
\makeatother

\newcommand{\beq}{\begin{equation}}
\newcommand{\eeq}{\end{equation}}
\newcommand{\ben}{\begin{eqnarray}}
\newcommand{\een}{\end{eqnarray}}
\newcommand{\beno}{\begin{eqnarray*}}
\newcommand{\eeno}{\end{eqnarray*}}

\numberwithin{equation}{section}

\begin{document}
	
\title[Instability of periodic traveling waves in the mKdV equation]{Instability bands for periodic traveling waves \\in the modified Korteweg-de Vries equation}

	\author{Shikun~Cui}
\address[Shikun~Cui]{School of Mathematical Sciences, Dalian University of Technology, Dalian, 116024,  China; 
	Department of Mathematics and Statistics, McMaster University, Hamilton, Ontario L8S 4K1, Canada}
\email{cskmath@163.com}
\author{Dmitry~E.~Pelinovsky}
\address[Dmitry~E.~Pelinovsky]{Department of Mathematics and Statistics, McMaster University, Hamilton, Ontario L8S 4K1, Canada }
\email{dmpeli@math.mcmaster.ca}

\date{\today}
\maketitle

\begin{abstract}
Two families of periodic traveling waves exist in the focusing 
mKdV (modified Korteweg-de Vries) equation. Spectral stability of these waveforms with respect to co-periodic perturbations of the same period has been previously explored by using spectral analysis and variational formulation. By using tools of integrability such as a relation between squared eigenfunctions of the Lax pair and eigenfunctions of the linearized stability problem, we revisit the spectral stability of these waveforms with respect to perturbations of arbitrary periods. In agreement with previous works, we find that one family is spectrally stable for all parameter configurations, whereas the other family is  spectrally unstable for all parameter configurations. We show that the onset of the co-periodic instability for the latter family changes the instability bands from figure-$8$ (crossing at the imaginary axis) into figure-$\infty$ (crossing at the real axis). 
\end{abstract}


\section{Introduction}

Instabilities of steadily propagating waves on a fluid surface, called Stokes waves, have been recently explored in many computational details due to advanced numerical algorithms with high precision and accuracy \cite{DDLS24,DO11,DS23,DS24,Lush2}. As the Stokes wave gets a larger height and higher steepness, the spectral problem for co-periodic perturbations admits more unstable eigenvalues which bifurcate from the origin due to coalescence of pairs of purely imaginary eigenvalues and splitting into pairs of real eigenvalues \cite{DS23}. At the same time, the spectral problem for perturbations of arbitrary periods display more complicated patterns of instability bands with multiple loops on the purely imaginary and real axes \cite{DO11,DS24}. The figure-$8$ instability (crossing at the imaginary axis) which is standard in the limit of small amplitudes \cite{Berti1,Berti2,Berti3,CD} transforms into the figure-$\infty$ instability (crossing at the real axis) with further exchanges between co-periodic and anti-periodic eigenvalues as eigenvalues of the dominant instability \cite{DDLS24}. It was conjectured that the recurrent exchange between these patterns is universal near the limit to the periodic wave with the maximal height and occur in other wave models such as the Whitham equation \cite{Carter}. 

This motivating picture related to the Stokes waves in Euler's equations calls for a more systematic investigation of the spectral instability of periodic traveling waves in the basic fluid models such as the focusing mKdV (modified Korteweg-de Vries) equation. The spectral stability of periodic traveling waves in this model has been studied in many details by using tools of integrability  \cite{CP0,CP1,DN1,Hoefer,Grava,LS1,LS2} and functional analysis \cite{angulo1,AN1,bronski,BHJ,DK,LP3,NLP2,pelinovsky}. The purpose of this work is to give a complete picture of the spectral stability of the general periodic traveling waves to perturbations of all periods and to show that the onset of the co-periodic instability changes the instability bands from figure-$8$ to figure-$\infty$.

We will describe the state-of-the-art and present the main results by using the standard form of the focusing mKdV equation,
\begin{equation}\label{ini_1}
u_t+6u^2u_x+u_{xxx}=0,
\end{equation}
where subscripts $x$ and $t$ represent the partial derivatives in the spatial and temporal variables respectively and where $u = u(x,t)$ is real. 
The initial-value problem for the mKdV equation (\ref{ini_1}) is globally well-posed in the energy space $H^1$ both on $\mathbb{R}$ and in the periodic domain \cite{KPV1}. It was further proven in \cite{CKSTT} that the global well-posedness can be extended to Sobolev spaces of low regularity in 
$H^s$ for $s>\frac{1}{4}$ on $\mathbb{R}$ and $H^s$ for $s>\frac{1}{2}$ in the periodic domain. By using integrability, the global well-posedness of the mKdV equation was extended to $H^s$ for $s > -\frac{1}{2}$ on $\mathbb{R}$ \cite{Visan}.

The travelling wave of the mKdV equation (\ref{ini_1}) with the spatial profile $U(x) : \mathbb{R}\to \mathbb{R}$ is given by $u(x,t)=U(x-ct)$, where $c$ is the wave speed. The wave profile $U$ satisfies the third-order equation 
\begin{equation}
\label{ini_3}
U'''+6U^2U'-cU'=0, 
\end{equation}
which is integrated into the second-order equation 
\begin{equation}\label{ini_4}
U''+2U^3-cU=b,
\end{equation}
with the integration constant $b$. Two particular waveforms for the periodic traveling waves were studied in many details:
\begin{equation}
\label{dn-wave}
U(x) = {\rm dn}(x,k), \quad c = 2-k^2, \quad b = 0
\end{equation}
and 
\begin{equation}
\label{cn-wave}
U(x) = k{\rm cn}(x,k), \quad c = 2k^2-1, \quad b = 0,
\end{equation}
where $k$ is elliptic modulus, $k \in(0,1)$, and the Jacobi elliptic functions have been used. We shall refer to (\ref{dn-wave}) and (\ref{cn-wave})  
as to the dnoidal and cnoidal waves respectively. With the use of the scaling transformation 
$$
u(x,t) \mapsto \alpha u(\alpha x, \alpha^3 t), \qquad \alpha > 0
$$
and the translational symmetries
$$
u(x,t) \mapsto u(x+x_0,t+t_0), \qquad x_0, t_0 \in \mathbb{R},
$$ 
these two waveforms extend to all possible traveling periodic wave solutions of the second-order equation (\ref{ini_4}) with $b = 0$.

It was shown in \cite{angulo1,AN1,DK,pelinovsky} that 
\begin{itemize}
	\item the dnoidal wave (\ref{dn-wave}) is spectrally stable to co-periodic perturbations for all $k \in (0,1)$, 
	\item the cnoidal wave (\ref{cn-wave}) is spectrally stable to co-periodic perturbations for $k \in (0,k_*)$ and spectrally unstable for $k \in (k_*,1)$, where $k_* \approx 0.909$. 
\end{itemize}
Furthermore, it was shown in \cite{bronski,BHJ} that 
\begin{itemize}
	\item the dnoidal wave (\ref{dn-wave}) is modulationally stable to long periods for all $k \in (0,1)$,
	\item the cnoidal wave (\ref{cn-wave}) is modulationally unstable to long periods for all $k \in (0,1)$. 
\end{itemize}

Regarding the periodic traveling waves with $b \neq 0$, stability of one particular solution was proven in \cite{AP2}. Two continuous families exist 
for $b \neq 0$ \cite{CP1,Hoefer} which extend the dnoidal and cnoidal waves (\ref{dn-wave}) and (\ref{cn-wave}), see equations (\ref{solution_1}) and (\ref{solution_2}) below. Stability of these waveforms with respect to co-periodic perturbations has been systematically studied in \cite{LP3,NLP2}
by using minimization of the quadratic form 
$$
\oint [(u')^2 + c u^2 ] dx
$$
subject to fixed mean value $\oint u dx = m$ and the fixed $L^4$ norm 
$\oint u^4 dx = 1$ (see review in \cite{Pelin-review}). 

\begin{itemize}
	\item Under the zero-mean constraint, $m = 0$, the cnoidal wave (\ref{cn-wave}) is a minimizer of the constrained variational problem for $k \in (0,k_*)$ and a saddle point for $k \in (k_*,1)$ with two symmetric minimizers bifurcating at $k = k_*$ in the supercritical pitchfork bifurcation \cite{NLP2}. \\
	
	\item Under a fixed nonzero mean value, $m \neq 0$, the imperfect pitchfork bifurcation breaks one branch of global minimizers of the constrained variational problem away from the other two branches, one of which is a saddle point and the other one is a local minimizer \cite{LP3}. The local and global minimizers are spectrally stable to co-periodic perturbations and the saddle points are spectrally unstable for all parameter values \cite{LP3}. Furthermore, the local and global minimizers are realized with both solution waveforms (\ref{solution_1}) and (\ref{solution_2}), whereas the saddle points are only realized by the waveform  (\ref{solution_2}) \cite{LP3}. 
\end{itemize}

Next, we present the main results of this work in relation to the state-of-the-art. 

\begin{itemize}
	\item First, we show by using the relation between squared eigenfunctions satisfying the Lax pair and eigenfunctions of the spectral stability problem that the solution waveform (\ref{solution_1}) generalizing dnoidal waves (\ref{dn-wave}) is spectrally stable for all parameter configurations and the solution waveform (\ref{solution_2}) generalizing cnoidal waves (\ref{cn-wave}) is spectrally unstable for all parameter configurations. This agrees with Theorem 2 in \cite{bronski} where modulational stability of the solution waveforms was studied to perturbations of long periods.\\
	
	\item Second, we show that the spectral instability of the solution waveform (\ref{solution_2}) to co-periodic perturbations affects the instability bands and triggers a transformation of figure-$8$ before the co-periodic instability to figure-$\infty$ after the co-periodic instability. This conclusion agrees with the modulational stability theory in \cite{bronski} and numerical approximations (see Figure 1 in \cite{bronski}) but has not been described in precise details. The analytical study of the co-periodic stability in \cite{DK} was accompanied by numerical approximations (see Figure 1 in \cite{DK}), which showed some transformations of figure-$8$ before figure-$\infty$ was attained. The stability spectrum of the dnoidal and cnoidal waves was computed numerically to in \cite{DN1} for just one value of $k \in (0,1)$, for which only figure-$8$ was obtained and the transformation to figure-$\infty$ was missed.
\end{itemize}

\begin{figure}[htpb!]
	\centering
	\subfigure[$k=0.75$]{\includegraphics[width=1.8in,height=1.5in]{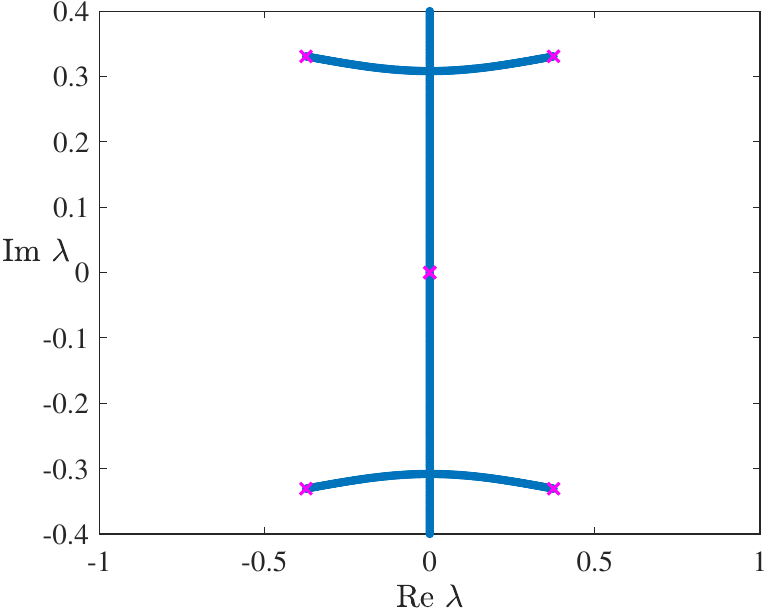}}
	\subfigure[$k=0.9$]{\includegraphics[width=1.8in,height=1.5in]{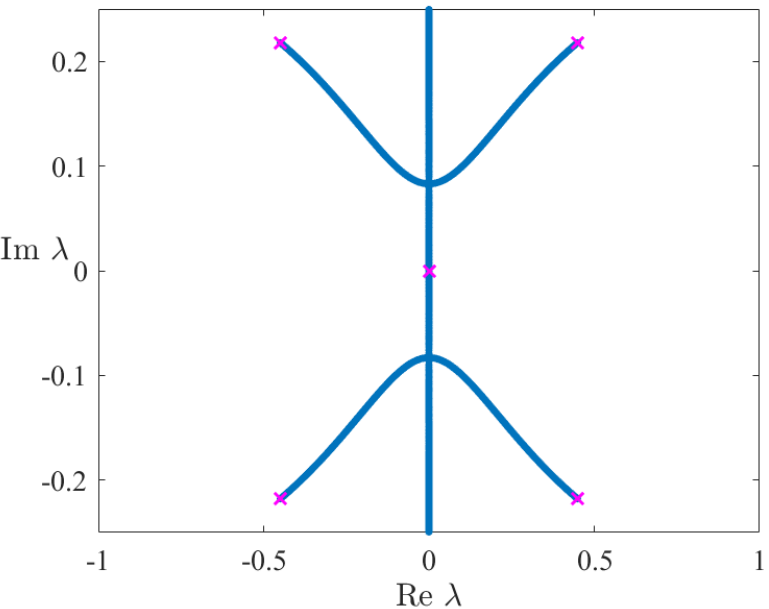}}
	\subfigure[$k=0.91$ ]{\includegraphics[width=1.8in,height=1.5in]{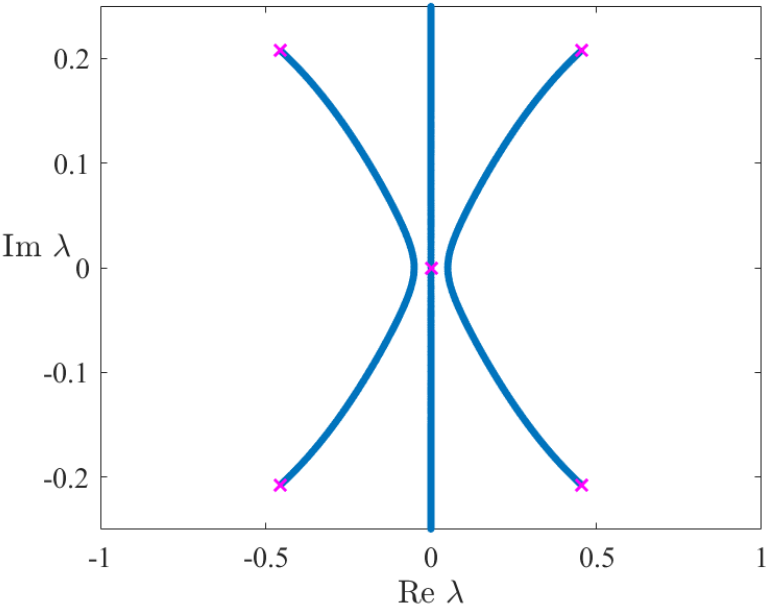}}\\
	\subfigure[$k=0.75$]{\includegraphics[width=1.8in,height=1.5in]{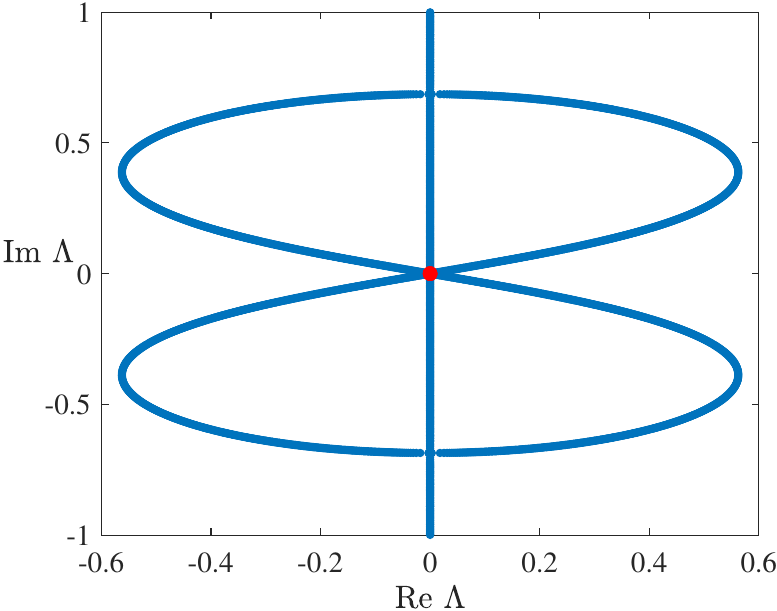}}
	\subfigure[$k=0.9$]{\includegraphics[width=1.8in,height=1.5in]{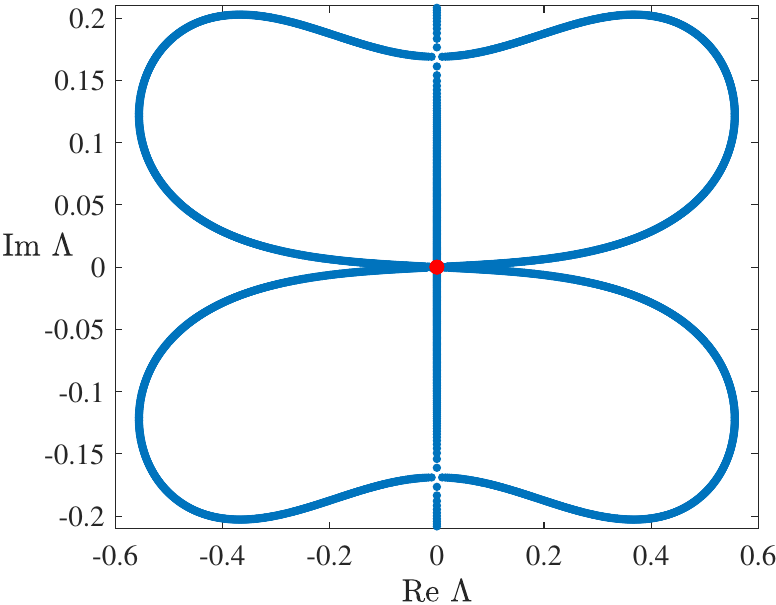}}
	\subfigure[$k=0.91$]{\includegraphics[width=1.8in,height=1.5in]{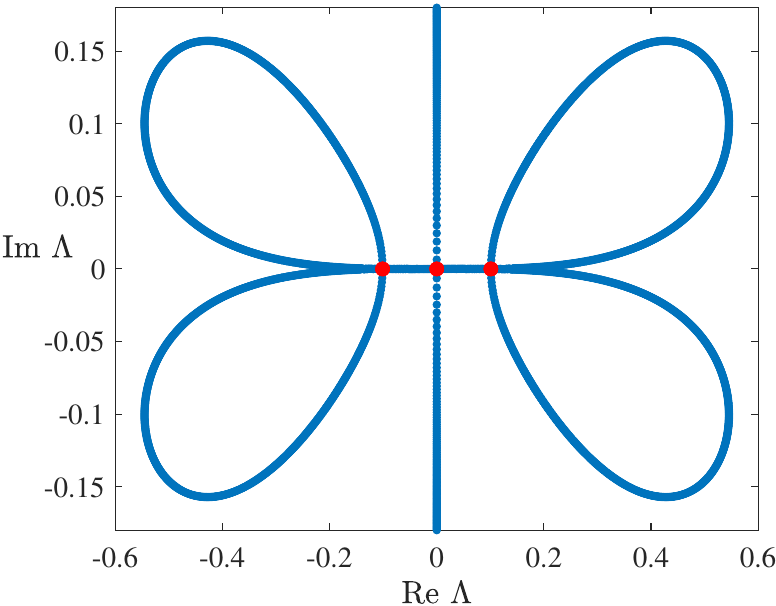}}\\
	\subfigure[$k=0.93$]{\includegraphics[width=1.8in,height=1.5in]{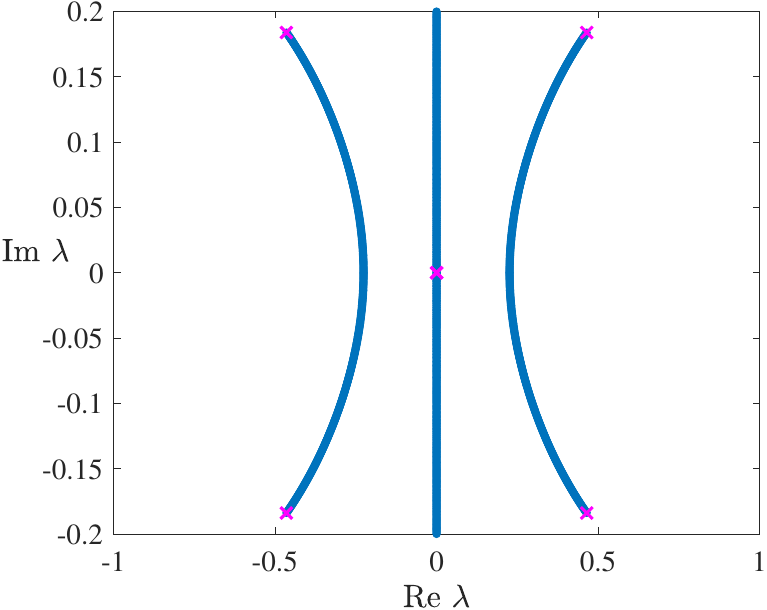}}
	\subfigure[$k=0.95$]{\includegraphics[width=1.8in,height=1.5in]{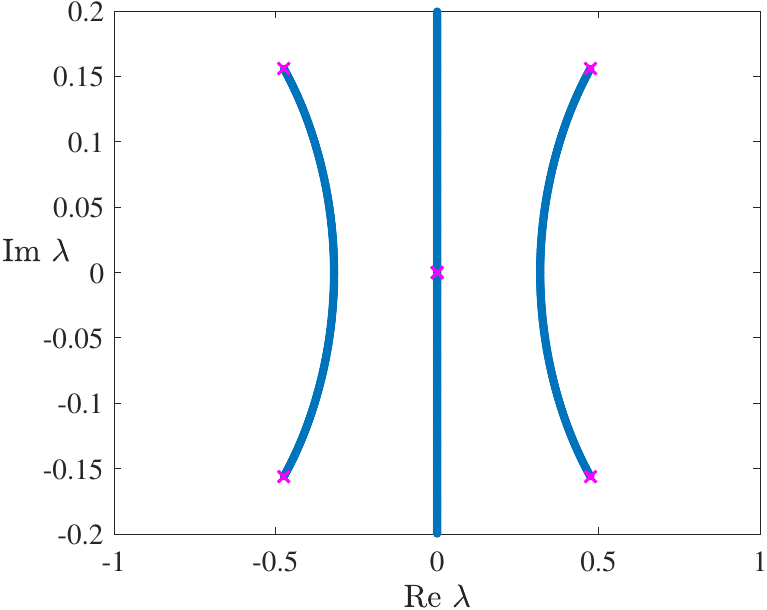}}	\subfigure[$k=0.97$]{\includegraphics[width=1.8in,height=1.5in]{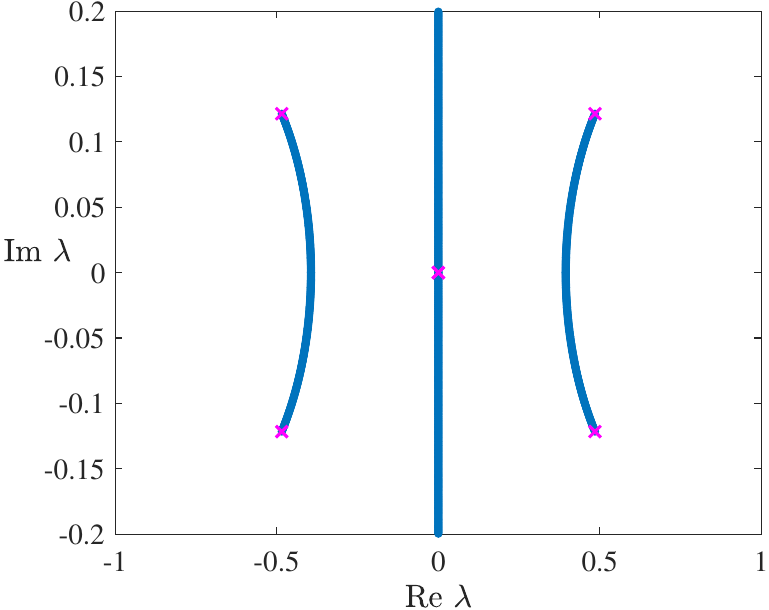}}\\
	\subfigure[$k=0.93$]{\includegraphics[width=1.8in,height=1.5in]{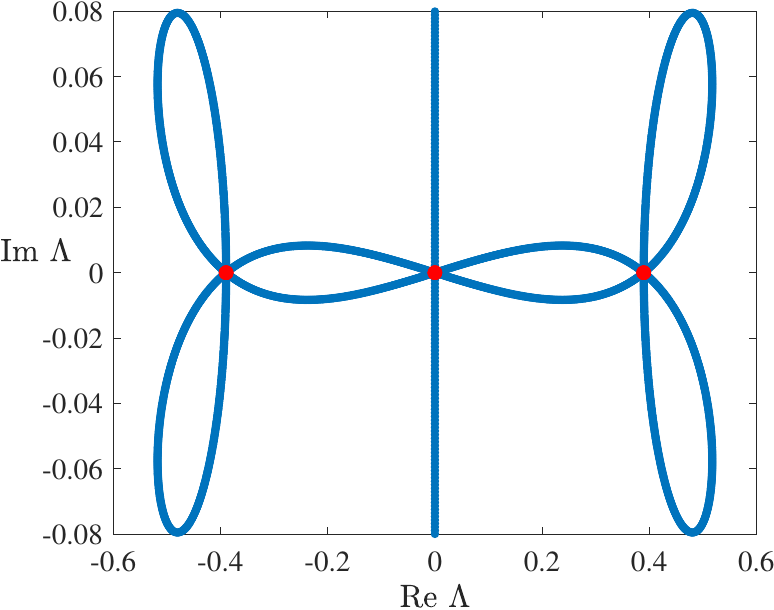}}	\subfigure[$k=0.95$]{\includegraphics[width=1.8in,height=1.5in]{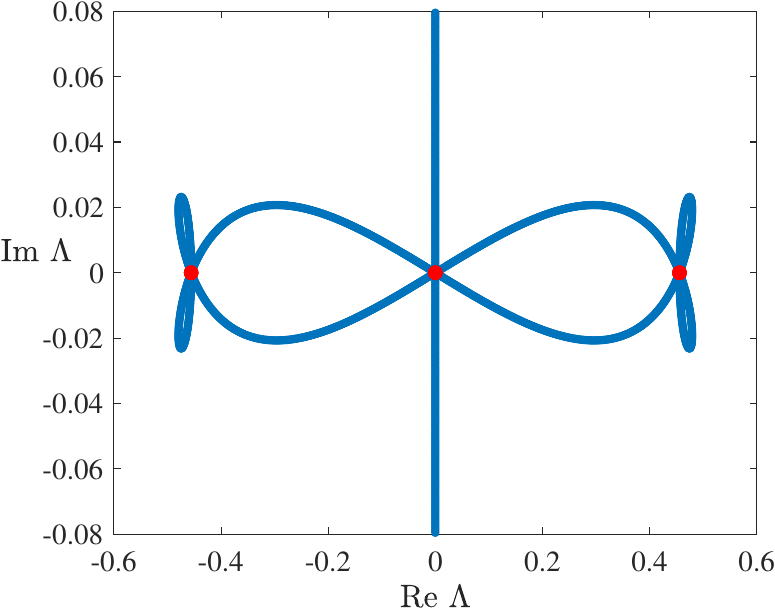}}
	\subfigure[$k=0.97$]{\includegraphics[width=1.8in,height=1.5in]{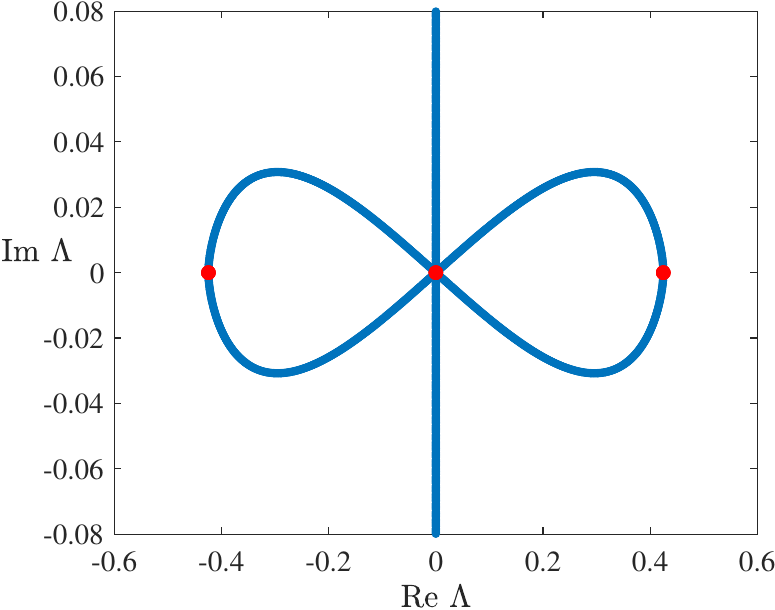}}
	\vspace{-0.2cm}
	\caption{The Lax and stability spectra for the cnoidal wave (\ref{cn-wave}) with different values of $k$. (a)-(c) and (g)-(i): Lax spectrum  in $\lambda$-plane. (d)-(f) and (j)-(l): stability spectrum  in $\Lambda$-plane.}
	\vspace{-0.2cm}
	\label{fig_case}
\end{figure}

The main phenomenon of the transformation of figure-$8$ to figure-$\infty$ is  shown on Figure \ref{fig_case} for the cnoidal wave (\ref{cn-wave}) with different values of $k \in (0,1)$. As $k$ crosses $k_* \approx 0.909$, the co-periodic instability arises [the corresponding eigenvalues are shown by red dots on panels (f) and (j)--(l)]. This leads to the transformation of the figure-$8$ [panels (d) and (e)] into a propeller [panel (f)], which becomes a complicated figure with two loops extended along the horizontal direction and four loops extended in the vertical direction [panels (j) and (k)]. Finally, the four loops disappear and the figure-$\infty$ is formed [panel (l)]. The point $k = k_*$ corresponds to the marginal case when the two bands of the Lax spectrum cross at the origin, separating their crossing at the pure imaginary axis for $k < k_*$ [panels (a) and (b)] and at the real axis for $k > k_*$ [panels (c) and (g)--(i)]. Four end points of the two complex bands of the Lax spectrum are shown by the magenta crosses, these are roots of the characteristic polynomial for the periodic traveling waves, see equation (\ref{p_1}) below.

Our analytical study and numerical computations of the spectral stability of the periodic traveling waves of the mKdV equation (\ref{ini_1}) rely on a relation between squared eigenfunctions satisfying the Lax pair and eigenfunctions of the spectral stability problem, see equation (\ref{Lambda}) below. The method of squared eigenfunctions was explored in a similar context of stability of periodic traveling waves in various integrable models in \cite{DMS,DS1,DU_1,UD_1}. The spectral stability problem can be solved due to separation of variables and this has been explored 
for integrable models in  \cite{CPU_1,CPW_1,CP_3,CuiP}, see also \cite{P_2,CP_2} for the cases where the spectral stability of periodic traveling waves cannot be studied by separation of variables.

It is interesting to emphasize that the Lax spectrum for the cnoidal wave (\ref{cn-wave}) of the focusing mKdV equation (\ref{ini_1}) is identical to that of the cnoidal wave solution of the focusing NLS equation \cite{DS1,CPW_1}. The spectral stability spectrum is however very different since the stability spectrum of the cnoidal waves in the focusing NLS equation only features the figure-$8$ instability \cite{DS1,CPW_1}. Although transformations of the instability bands for Stokes waves in Euler's equations are more complicated than the one in Figure \ref{fig_case}, see \cite{DDLS24}, the main transformation of figure-$8$ into figure-$\infty$ due to the co-periodic instability is well captured by the wave model of the focusing mKdV equation. 

The paper is organized as follows. Section \ref{sec-2} presents the waveforms for the periodic traveling waves of the mKdV equation (\ref{ini_1}). Section \ref{sec_eig} characterizes the periodic traveling waves by using the Lax pair and gives a relation between the squared eigenfunctions of the Lax pair and the eigenfunctions of the linearized mKdV equation. Section \ref{sec_laxpspec} 
describes the analytical results on the location of the Lax and stability spectra for the two waveforms of the periodic traveling waves. 
Section \ref{sec_lax_periodic} contains outcomes of the numerical computations based on a robust numerical algorithm for approximation of the Lax spectrum. Section \ref{sec_modu} gives a summary of our findings.

\section{Waveforms for the periodic traveling waves}
\label{sec-2}

Integrating the second-order equation (\ref{ini_4}), we obtain the first-order invariant
\begin{equation}
\label{ini_5}
(U')^2 + Q(U) = d,
\end{equation}
where $Q(U):= U^4-cU^2-2bU$ and $d$ is a constant. Two particular waveforms generalizing the dnoidal and cnoidal waves (\ref{dn-wave}) and (\ref{cn-wave}) are given in terms of the four roots $\{ u_1,u_2,u_3.u_4 \}$ of $Q(u)=d$. The four roots satisfy the relations
	\begin{equation}\label{relation_u1}
	\begin{cases}
	u_1+u_2+u_3+u_4=0,\\
    u_1u_2+u_1u_3+u_1u_4+u_2u_3+u_2u_4+u_3u_4 = -c,\\
	u_1u_2u_3+u_1u_2u_4+u_1u_3u_4+u_2u_3u_4=2b,\\
	u_1u_2u_3u_4=-d,
	\end{cases}
	\end{equation}
which follow by expanding 
	\begin{align*}
	Q(u)-d = (u-u_1)(u-u_2)(u-u_3)(u-u_4) = u^4 - c u^2 - 2 b u - d.
	\end{align*}
The following proposition states the existence of two waveforms for the periodic traveling waves. The statement is based on the explicit computations verified in \cite{CP1}. 

\begin{proposition}
	\label{pro_periodic}
If the roots $\{ u_1,u_2,u_3.u_4 \}$ are real and ordered as $u_4\leq u_3\leq u_2\leq u_1$,  then the first-order invariant (\ref{ini_5}) is satisfied by 
	\begin{equation}
	\label{solution_1}
	U(x) = u_4+\frac{(u_1-u_4)(u_2-u_4)}{(u_2-u_4)+(u_1-u_2){\rm sn}^2(\nu x,k)},
	\end{equation}
	where 
	$$
	\nu = \frac{1}{2}\sqrt{(u_1-u_3)(u_2-u_4)}, \quad 
	k = \frac{\sqrt{(u_1-u_2)(u_3-u_4)}}{\sqrt{(u_1-u_3)(u_2-u_4)}}.
	$$
If $\{ u_1,u_2 \}$ are real and $\{u_3.u_4 \}$ are complex-conjugate 
such that $u_2\leq u_1$ and $u_3=\bar{u}_4=\gamma+i\eta$ with $\eta > 0$,
then the first-order invariant (\ref{ini_5}) is satisfied by 
	\begin{equation}
	\label{solution_2}
	{U}(x) = u_2 + \frac{({u}_1-{u}_2)\big(1-{\rm cn}(\mu x,k)\big)}{(\delta+1)+(\delta-1){\rm cn}(\mu x,k)},
	\end{equation}
	where 
	\begin{align*}
	\delta = \frac{\sqrt{({u}_1-\gamma)^2+\eta^2}}{\sqrt{({u}_2-\gamma)^2+\eta^2}},  \quad 
	\mu = \sqrt[4]{[({u}_1-\gamma)^2+\eta^2][({u}_2-\gamma)^2+\eta^2]},
	\end{align*}
	and
	\begin{align*}
	2k^2&=1-\frac{({u}_1-\gamma)({u}_2-\gamma)+\eta^2}{\sqrt{[({u}_1-\gamma)^2+\eta^2][({u}_2-\gamma)^2+\eta^2]}}.
	\end{align*}
\end{proposition}

\begin{remark}
	For periodic solution (\ref{solution_1}), exchanging $u_1 \leftrightarrow u_3$ and  $u_2 \leftrightarrow u_4$ yields another periodic solution:
	\begin{equation}\label{solution_3}
	U(x)=u_2+\frac{(u_2-u_3)(u_2-u_4)}{(u_4-u_2)+(u_3-u_4){\rm sn}^2(\nu x;k)},
	\end{equation}
with the same definition of $\nu$ and $k$.
\end{remark}

\begin{remark}
	The periodic solutions $U(x)$ in the form of (\ref{solution_1}) and (\ref{solution_3}) are located in the intervals $[u_2,u_1]$ and $[u_4,u_3]$, respectively. In both cases, they have the period $L=2K(k)\nu^{-1}$, where $K(k)$ is the complete elliptic integral of the first type.	The periodic solution ${U}(x)$ in the form of (\ref{solution_2}) is located in the interval $[{u}_2,{u}_1]$ and has the period $L=4K(k)\mu^{-1}$. 
\end{remark}

To illustrate the two waveforms in Proposition \ref{pro_periodic}, 
we plot the level curves of 
$$
\mathcal{H}(U,U') = (U')^2 + Q(U)
$$ 
on the phase plane $(U,U') \in \R^2$. Figure \ref{fig_phase}(a) shows a typical phase portrait in the case $c > 0$ and $b \in \left(-\frac{\sqrt{2c^3}}{3 \sqrt{3}},\frac{\sqrt{2c^3}}{3 \sqrt{3}}\right)$, where the saddle point is squeezed between two center points shown by black crosses. The level curves on the phase plane $(U,U')$ 
correspond to the region of $U$, where $Q(U) \leq d$, shown in Figure \ref{fig_phase}(b).
\begin{figure}[h]
	\centering
	\includegraphics[width=3in,height=2.5in]{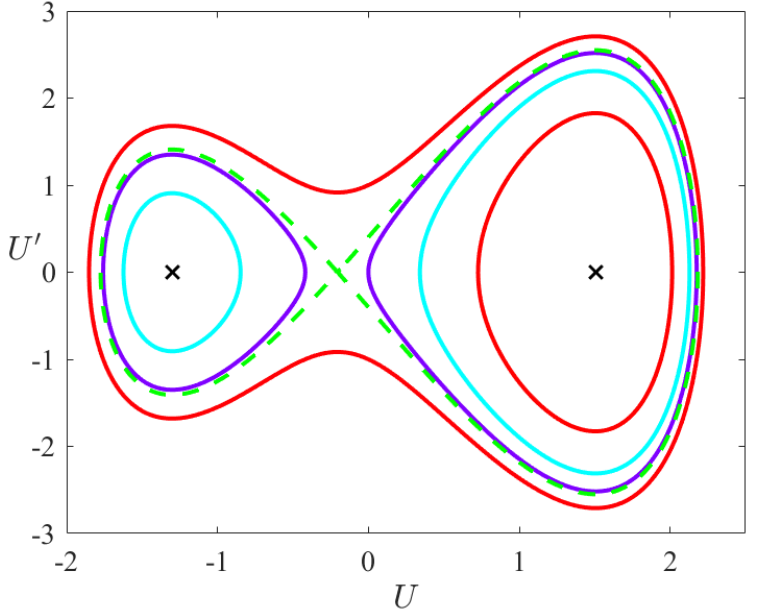}
	\includegraphics[width=3in,height=2.5in]{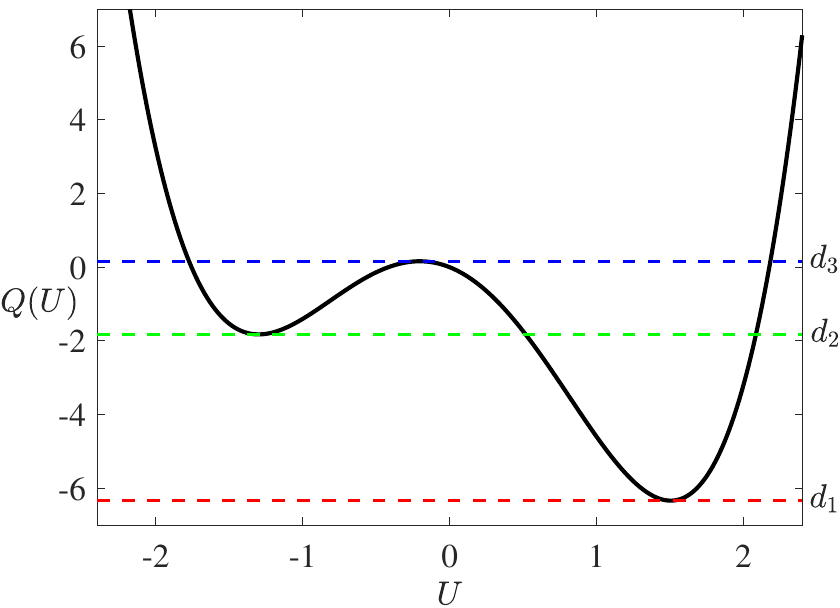}
	\caption{(a) Phase portrait in the phase plane $(U,U')$  for $b = 0.8$ and $c = 4$. (b) Levels of $d$ at the plot of $Q(U)$.}
	\label{fig_phase}
\end{figure}

For $d \in (d_1,d_2)$ and $d \in (d_3,\infty)$, only two real roots 
of $Q(U) = d$ exist and the solution waveform is given by (\ref{solution_2}). 
For $d \in (d_2,d_3)$, four real roots of $Q(U) = d$ exist 
and the solution waveforms are given by (\ref{solution_1}) and (\ref{solution_3}). In the limiting cases, the solution waveforms degenerate as follows. 

\begin{itemize}
	\item If $d = d_1$, then $u_1 = u_2$ in (\ref{solution_2}) and we have the constant solution $U(x) = u_1$.
	
	\item If $d = d_2$, then $u_3 = u_4$ in (\ref{solution_1}) and we have 
	the trigonometric waveform for the periodic solution:
	$$	
	U(x)=u_3+\frac{(u_1-u_3)(u_2-u_3)}{(u_2-u_3)+(u_1-u_2) \sin^2(\nu x)}, \quad \nu=\frac{1}{2}\sqrt{(u_1-u_3)(u_2-u_3)}.
	$$
In addition, we have the constant solution $U(x)=u_3$ from (\ref{solution_3}). 
	
	\item If $d=d_3$, then either $u_2=u_3$ in (\ref{solution_1}) or $\eta = 0$ and $u_2 < \gamma < u_1$ in (\ref{solution_2}). The two cases give two hyperbolic waveforms for the solitary wave solutions:
	$$
	U(x)=u_4+\frac{(u_1-u_4)(u_2-u_4)}{(u_2-u_4)+(u_1-u_2) \tanh^2(\nu x)}, \quad \nu=\frac{1}{2}\sqrt{(u_1-u_2)(u_2-u_4)}
	$$
and 
\begin{align*}
U(x) &= u_4 + \frac{(u_1 - u_4) (u_2 - u_4) (\cosh(2 \nu x) - 1)}{(u_1 - u_4) \cosh(2 \nu x) + (u_1 - 2 u_2 + u_4)} \\
&= u_4 + \frac{(u_1 - u_4) (u_2 - u_4) \tanh^2(\nu x) }{(u_1 - u_2) + (u_2 - u_4) \tanh^2(\nu x)},
\end{align*}
where we have relabeled $u_2 \to u_4$ and $\gamma \to u_2 = u_3$ in the second solution compared to (\ref{solution_2}) and transformed it with $\mu = 2 \nu$ to the form similar to the first solution. In addition, we have the constant solution $U(x)=u_2$ from (\ref{solution_3}). 
\end{itemize}

\section{Lax spectrum for the periodic traveling waves}
\label{sec_eig}

The mKdV equation (\ref{ini_1}) is a compatibility condition of the following system of the linear equations \cite{AKNS1} for $\psi \in \mathbb{C}^2$:
\begin{equation}\label{lax_1}
\left\{
\begin{array}{lr}
\psi_{x}=L(u,\lambda)\psi, \\
\psi_{t}=M(u,\lambda)\psi,
\end{array}
\right.
\end{equation}
where
$$
L(u,\lambda) = \left(
\begin{array}{cc}
\lambda       &u  \\
-u & -\lambda \\
\end{array}
\right)
$$
and 
$$
M(u,\lambda) = \left(
\begin{array}{cc}
-4\lambda^3-2\lambda u^2  &-4\lambda^2u-2\lambda u_x-2u^3-u_{xx}  \\
4\lambda^2u-2\lambda u_x+2u^3+u_{xx} & 4\lambda^3+2\lambda u^2 \\
\end{array}
\right).
$$
Solutions of the linear equations (\ref{lax_1}) for the traveling waves of the mKdV equation (\ref{ini_1}) are related to solutions of the linearized mKdV equation at the traveling waves. These relations are well-known, see, e.g., \cite{DN1}, and are reproduced here for the sake of transparency.

Let $u(x,t) = U(x-ct)$ be the traveling wave solution of the mKdV equation (\ref{ini_1}). Then the linear system (\ref{lax_1}) enjoys the separation of  variables in the form $\psi(x,t) = \Psi(x-ct) e^{\Omega t}$, with $\Psi \in \mathbb{C}^2$ and $\Omega \in \mathbb{C}$ found from the linear system
\begin{equation}
\begin{cases}\label{eig_1}
\Psi'=L(U,\lambda)\Psi, \\
\Omega\Psi = \left[ M(U,\lambda) + c L(U,\lambda) \right] \Psi.
\end{cases}
\end{equation}
The following proposition gives the admissible values of $\Omega$.

\begin{proposition}\label{pro_omega_1}
	The admissible values of $\Omega$ are defined by the characteristic polynomial 
\begin{equation}\label{p_1}
P(\lambda)=16\lambda^6-8c\lambda^4+(c^2+4d)\lambda^2-b^2
\end{equation}
	as $\Omega=\pm \sqrt{P(\lambda)}$.
\end{proposition}

\begin{proof}
The second equation in system (\ref{eig_1}) shows that $\Omega$ is an eigenvalue of the following matrix 
	$$
A := \left(
	\begin{array}{cc}
	-4\lambda^3-2\lambda U^2+c\lambda       &-(4\lambda^2U+2U^3+U''-cU+2\lambda U')  \\
	4\lambda^2U+2U^3+U''-cU-2\lambda U' 
	&4\lambda^3+2\lambda U^2-c\lambda \\
	\end{array}
	\right).
	$$
Since the trace of $A$ is zero, $\Omega^2 = P(\lambda) := \det(A)$, which is expanded in powers of $\lambda$ as follows:
	\begin{align*}
P(\lambda) = 16\lambda^6 -8c\lambda^4 -(12U^4+8UU''-4cU^2-4(U')^2-c^2)\lambda^2
-(2U^3+U''-cU)^2.
	\end{align*}
By using (\ref{ini_4}) and (\ref{ini_5}), we rewrite $P(\lambda)$ in the form (\ref{p_1}). 	
\end{proof}

We define {\it the Lax spectrum} of the periodic travelling waves as the set of admissible values of $\lambda$ in the spectral problem $\Psi' = L(U,\lambda) \Psi$ of the system (\ref{eig_1}), for which $\Psi \in L^{\infty}(\mathbb{R},\mathbb{C}^2)$. By Floquet theorem, if $U(x+L) = U(x)$ is $L$-periodic, then the bounded eigenfunction $\Psi$ is quasi-periodic 
as $\Psi(x+L) = \Psi(x) e^{i \mu L}$ with $\mu \in \left[-\frac{\pi}{L},\frac{\pi}{L}\right]$ for the values of $\lambda$ defined in the continuous bands of the Lax spectrum. 

To relate the Lax spectrum with the stability spectrum, we define the linearization of the mKdV equation (\ref{ini_1}) at the periodic traveling waves with the profile $U$. By using $u(x,t) = U(x-ct) + \mathfrak{u}(x-ct) e^{\Lambda t}$ and linearizing at $\mathfrak{u}$, we obtain the spectral stability problem in the form
\begin{equation}\label{line_11}
\Lambda \mathfrak{u} + \mathfrak{u}''' + 6 (U^2 \mathfrak{u})' - c \mathfrak{u}' = 0. 
\end{equation} 
{\it The stability spectrum} of the periodic traveling waves is defined as the set of admissible values of $\Lambda$ in the spectral stability problem (\ref{line_11}), for which $\mathfrak{u} \in L^{\infty}(\mathbb{R},\mathbb{C})$. By the same Floquet theorem, if $U(x+L) = U(x)$ is $L$-periodic, then the bounded eigenfunction $\mathfrak{u}$ is quasi-periodic as $\mathfrak{u}(x+L) = \mathfrak{u}(x) e^{i \theta L}$ with $\theta \in \left[-\frac{\pi}{L},\frac{\pi}{L}\right]$  for the values of $\Lambda$ defined in the continuous bands of the stability spectrum. The bands of $\Lambda$ in the stability spectrum are related to the bands of $\lambda$ in the Lax spectrum due to the squared eigenfunction relation \cite{DN1}.

The following proposition gives the  relation between eigenfunctions of the spectral stability problem (\ref{line_11}) and the squared eigenfunctions of the linear system (\ref{eig_1}). 

\begin{proposition}
	Let $\Sigma \subset \mathbb{C}$ be the Lax spectrum and 
	$\Psi=(p,q)^{\mathrm{T}} \in L^{\infty}(\mathbb{R},\mathbb{C}^2)$ be the eigenfunction of the linear system (\ref{eig_1}) for an admissible value of $\lambda \in \Sigma$. Then, 
	$\mathfrak{u} := p^2 - q^2 \in L^{\infty}(\mathbb{R},\mathbb{C})$ is the eigenfunction of the spectral stability problem (\ref{line_11}) with 
\begin{equation}
\label{Lambda}
	\Lambda = 2\Omega = \pm 2\sqrt{P(\lambda)}.
\end{equation}
\label{prop-squared}
\end{proposition}

\begin{proof}
	Given that 	$\Psi=(p,q)^{\mathrm{T}}$, we rewrite $\Psi'=L(U,\lambda)\Psi$ into 
	\begin{equation*}
	\begin{cases} 
	p'=\lambda p+Uq,\\
	q'=-Up-\lambda q,
	\end{cases}
	\end{equation*}
from which we obtain
	\begin{equation*}
	\begin{cases}
	p'' &=\lambda p'+U'q+Uq'\\
	&=\lambda^2 p + U'q-U^2p,\\
	q'' &=-U'P-Up'-\lambda q' \\
	&=\lambda^2 q-U^2q-U'p
	\end{cases}	
	\end{equation*}
and
	\begin{equation*}
	\begin{cases}
	p'''&=\lambda^2 p'+U''q+U'q'-2UU'p-U^2p'\\
	&=\lambda^3p+\lambda^2Uq+U''q-3UU'p-\lambda U'q-\lambda U^2p-U^3q, \\
	q'''&=\lambda^2 q'-U''p-U'p'-2UU'q-U^2q'\\
	&=-\lambda^3q-\lambda^2Up-U''p-3UU'q-\lambda U'p+\lambda U^2q+U^3p.
	\end{cases}	
	\end{equation*}
This yields with explicit computations:
	\begin{equation*}
	\begin{aligned}
	\quad(p^2)''' &= 2pp'''+6p'p''\\
	&=8\lambda^3 p^2+8\lambda^2 Upq+4\lambda U'pq-8\lambda U^2p^2
	-8 U^3pq+6UU'q^2-6UU'p^2+2U''pq
	\end{aligned}	
	\end{equation*}
	and
	\begin{equation*}
	\begin{aligned}
	\quad(q^2)''' &= 2pp'''+6p'p''\\
	&=-8\lambda^3 q^2-8\lambda^2 Upq+4\lambda U'pq+8\lambda U^2q^2
	+8 U^3pq-6UU'q^2+6UU'p^2-2U''pq.
	\end{aligned}	
	\end{equation*}
Combining with the second equation of (\ref{eig_1}), we obtain
	\begin{equation*}
	\begin{aligned}
	&\quad(p^2)'''+12UU'p^2+6U^2(p^2)'-c(p^2)'\\
	&=8\lambda^3 p^2+8\lambda^2 Upq+4\lambda U'pq
	+4\lambda U^2p^2+4U^3pq+2U''pq+6UU'q^2+6UU'p^2\\
	&=-2\Omega p^2+6UU'(p^2+q^2)
	\end{aligned}	
	\end{equation*}
	and 
	\begin{equation*}
	\begin{aligned}
	&\quad(q^2)'''+12UU'q^2+6U^2(q^2)'-c(q^2)'\\
	&=-8\lambda^3 q^2-8\lambda^2 Upq+4\lambda U'pq
	-4\lambda U^2q^2-4U^3pq-2U''pq+6UU'q^2+6UU'p^2\\
	&=-2\Omega q^2+6UU'(p^2+q^2),
	\end{aligned}	
	\end{equation*}
which verify that  
	\begin{equation*}
\mathfrak{u}''' + 12 UU' \mathfrak{u} + 6U^2\mathfrak{u}' - c \mathfrak{u}' = -2\Omega \mathfrak{u},
	\end{equation*} 
for $\mathfrak{u} := p^2 - q^2$. This equation is equivalent to (\ref{line_11}) with $\Lambda = 2 \Omega$, whereas the relation $\Omega = \pm \sqrt{P(\lambda)}$ is established in Proposition \ref{pro_omega_1}.
\end{proof}

\section{Lax and stability spectra: analytical results}
\label{sec_laxpspec}

Although the exact location of the Lax spectrum for the periodic traveling waves is not known, the symmetry of the Lax spectrum with respect to reflection about $\lambda = 0$ and about the real axis follow the symmetry of $L(u,\lambda)$ in the linear system (\ref{lax_1}). The following proposition states these properties of the Lax spectrum.

\begin{proposition}\label{proposition_11}
	Let $\lambda\in \mathbb{C}$ be an eigenvalue of the spectral problem $\psi_{x}=L(u,\lambda)\psi$ with the eigenfunction $\psi = (p,q)^T \in L^{\infty}(\mathbb{R},\mathbb{C}^2)$. Then, $-\lambda$ is also an eigenvalue with the eigenfunction $\psi = (q,-p)^T \in L^{\infty}(\mathbb{R},\mathbb{C}^2)$. Moreover, if $\lambda \notin \R$, then $\bar{\lambda}$ is also an eigenvalue with the eigenfunction $\psi = (\bar{p},\bar{q})^T \in L^{\infty}(\mathbb{R},\mathbb{C}^2)$. 
\end{proposition}

\begin{proof}
	We rewrite $\psi_{x}=L(u,\lambda)\psi$ in the form:
	\begin{align*}
	\left\{ \begin{array}{l} p_x=\lambda p+uq, \\
	q_x=-up-\lambda q.
	\end{array} \right.
	\end{align*}	
	This implies  
	\begin{align*}
	\left\{ \begin{array}{l} 
	q_x=(-\lambda)q+u(-p), \\
	-p_x=-\lambda p-uq,
	\end{array} \right.
	\end{align*}
	so that $\psi := (q,-p)^T$ is also a solution of   $\psi_{x}=L(u,\lambda)\psi$ with eigenvalue $-\lambda$.
	Furthermore, taking the complex-conjugate transform yields
	\begin{align*}
	\left\{ \begin{array}{l} 
	\bar{p}_x=\bar{\lambda} \bar{p}+\bar{u}\bar{q}, \\
	\bar{q}_x=-\bar{u}\bar{p}-\bar{\lambda}\bar{q}.
	\end{array} \right.
	\end{align*}
	Since $u=\bar{u}$ and $\lambda \neq \bar{\lambda}$, then 
	$\psi :=(\bar{p},\bar{q})^T$ is also a solution of  $\psi_{x}=L(u,\lambda)\psi$ with eigenvalue $\bar{\lambda}$.	
\end{proof}

Roots of the characteristic polynomial $P(\lambda)$ in (\ref{p_1}) can be enumerated as $\{ \pm\lambda_1, \pm\lambda_2, \pm\lambda_3 \}$. It was found in \cite{Kam1,Kam2,Kam3,Kam4} that roots of $P(\lambda)$ are related to roots $\{ u_1,u_2,u_3.u_4 \}$ of $Q(u) = d$ in (\ref{ini_5}). These relations were verified for the mKdV equation in \cite{CP1} with a direct proof, hence we state the relations without further details: 
	\begin{align}
	\label{rel-lambda-u}
	\lambda_1=\frac{1}{2}(u_1+u_2), \qquad 
	~\lambda_2=\frac{1}{2}(u_1+u_3), \qquad 
	~\lambda_3=\frac{1}{2}(u_2+u_3).
	\end{align}
It was shown in \cite{CP1} that spectral bands of the Lax spectrum for the periodic traveling waves outside $i \R$ are connected between the roots 
$\{ \pm\lambda_1, \pm\lambda_2, \pm\lambda_3 \}$. In the next two theorems,  we derive the stability criterion for the waveform (\ref{solution_1}) 
and the instability criterion for the waveform (\ref{solution_2}) in the spectral stability problem (\ref{line_11}). It relies on the location of the spectral bands between the roots of $P(\lambda)$ in addition to the spectral bands on $i \mathbb{R}$, which is assumed here and verified numerically in Section \ref{sec_lax_periodic}, as well as the squared eigenfunction relation of Proposition \ref{prop-squared}.

\begin{theorem}
	\label{theorem-stab}
For the waveform (\ref{solution_1}) of Proposition \ref{pro_periodic}, assume that the Lax spectrum is located on 
\begin{equation}
\label{Lax-dn}
	i \mathbb{R} \cup [-\lambda_1,-\lambda_2] \cup [-|\lambda_3|,|\lambda_3|] \cup [\lambda_2,\lambda_1].
\end{equation}
Then the stability spectrum is $i \mathbb{R}$.
\end{theorem}

\begin{proof}
	For the waveform (\ref{solution_1}) of Proposition \ref{pro_periodic}, the roots $\{ u_1,u_2,u_3.u_4 \}$ of $Q(u) = d$ are real and ordered as $u_4 \leq u_3 \leq u_2 \leq u_1$. By (\ref{rel-lambda-u}), the roots $\{ \pm\lambda_1, \pm\lambda_2, \pm\lambda_3 \}$ of $P(\lambda)$ are real and ordered as 
	$\lambda_3 \leq \lambda_2 \leq \lambda_1$. Moreover, due to the first relation in (\ref{relation_u1}), we have $u_2 + u_3 = -u_1 - u_4$. If $u_2 + u_3 < 0$, then $|u_2 + u_3| = u_1 + u_4 \leq u_1 + u_3$, which proves that $|\lambda_3| \leq \lambda_2$ even if $\lambda_3 < 0$. Hence, we have  $0 \leq |\lambda_3| \leq \lambda_2 \leq \lambda_1$. 

For $\lambda\in i\mathbb{R}$ in the Lax spectrum (\ref{Lax-dn}), we rewrite the polynomial $P(\lambda)$ in the form:
\begin{align*}
P(\lambda) &= 16(\lambda^2-\lambda_1^2)
	(\lambda^2-{\lambda_2^2})
	({\lambda^2}-{\lambda_3^2}) \\
&= -16 (|\lambda|^2 + \lambda_1^2)
	(|\lambda|^2 + {\lambda_2^2})
	(|\lambda|^2 + {\lambda_3^2}).
\end{align*} 
Since $P(\lambda) < 0$, we have 
$\Lambda = \pm 2\sqrt{P(\lambda)}\in i\mathbb{R}$ by (\ref{Lambda}). 
Moreover, $\lim\limits_{|\lambda| \to \infty} P(\lambda) = -\infty$ and 
$\lim\limits_{\lambda \to 0} P(\lambda) = - 16 \lambda_1^2 \lambda_2^2 \lambda_3^2$ so that 
$$
(-\infty,-8 \lambda_1 \lambda_2 |\lambda_3|] \cup [8 \lambda_1 \lambda_2 |\lambda_3|,\infty) i
$$
belongs to the stability spectrum.

For $\lambda \in [-\lambda_1,-\lambda_2] \cup [-|\lambda_3|,|\lambda_3|] \cup [\lambda_2,\lambda_1]$ in the Lax spectrum (\ref{Lax-dn}), we have
$P(\lambda) \leq 0$ so that $\Lambda = \pm 2\sqrt{P(\lambda)}\in i\mathbb{R}$ by (\ref{Lambda}). Moreover, 
	$\lim\limits_{|\lambda| \to |\lambda_3|} P(\lambda) = 0$ so that 
	$$
	[-8 \lambda_1 \lambda_2 |\lambda_3|,8 \lambda_1 \lambda_2 |\lambda_3|] i
	$$
	also belongs to the stability spectrum. Hence, the image of $\Lambda = \pm 2\sqrt{P(\lambda)}$ covers all $i \mathbb{R}$.	
\end{proof}

\begin{remark}
	The parts of the Lax spectrum (\ref{Lax-dn}) for $\lambda \in  [-\lambda_1,-\lambda_2] \cup  [\lambda_2,\lambda_1]$ cover a subset of $i \mathbb{R}$ for the second time. Regardless, the periodic traveling wave with the waveform (\ref{solution_1}) is classified as spectrally stable.
\end{remark}

\begin{theorem}
	\label{theorem-instab}
	For the waveform (\ref{solution_2}) of Proposition \ref{pro_periodic}, assume that the Lax spectrum is located on 
	\begin{equation}
	\label{Lax-cn}
	i \mathbb{R} \cup [-|\lambda_1|,|\lambda_1|] \cup \Sigma_+ \cup \Sigma_-,
	\end{equation}
where $\Sigma_+$ is the spectral band connecting two complex roots in the set $\{ \pm\lambda_2, \pm\lambda_3 \}$ and $\Sigma_-$ is the reflection of $\Sigma_+$ about $\lambda = 0$. Then the stability spectrum is the union of $i \mathbb{R}$ and the spectral bands which are not contained in $i \mathbb{R}$ and are symmetric about $\Lambda = 0$.
\end{theorem}

\begin{proof}
	For the waveform (\ref{solution_2}) of Proposition \ref{pro_periodic}, the roots $\{ u_1,u_2  \}$ of $Q(u) = d$ are real and ordered as $u_2 \leq u_1$, whereas the roots $\{ u_3, u_4\}$ are complex conjugate with $u_3 = \bar{u}_4 = \gamma + i \eta$ and $\eta > 0$. By (\ref{rel-lambda-u}), the roots $\{ \pm \lambda_1 \}$ of $P(\lambda)$ are real, whereas the roots 
	$\{ \pm\lambda_2, \pm\lambda_3 \}$ represent a complex-conjugate quadruplet since 
	$$
	\lambda_3 = \frac{1}{2} (u_2 + u_3) = -\frac{1}{2} (u_1 + u_4) = -\frac{1}{2} (u_1 + \bar{u}_3) = -\bar{\lambda}_2,
	$$
	due to the first relation in (\ref{relation_u1}). Since $\lambda_2 = -\bar{\lambda}_3$, we can rewrite $P(\lambda)$ in the form:
$$
P(\lambda)= 16 (\lambda^2-\lambda_1^2)
[(\lambda^2- {\rm Re}(\lambda_2^2))^2 + ({\rm Im}(\lambda_2^2))^2].
$$

For $\lambda\in i\mathbb{R}$ in the Lax spectrum (\ref{Lax-cn}), 
we have $P(\lambda) < 0$, and since $\lim\limits_{|\lambda| \to \infty} P(\lambda) = -\infty$ and 
$\lim\limits_{\lambda \to 0} P(\lambda) = - 16 \lambda_1^2 \lambda_2^2 \lambda_3^2$, we have by $\Lambda = \pm 2\sqrt{P(\lambda)}$ that 
$$
(-\infty,-8 |\lambda_1| |\lambda_2|^2 ] \cup [8 |\lambda_1| |\lambda_2|^2,\infty) i
$$
belongs to the stability spectrum.

For $\lambda \in [-|\lambda_1|,|\lambda_1|]$ in the Lax spectrum (\ref{Lax-cn}), we have $P(\lambda) \leq 0$ and since 
$\lim\limits_{|\lambda| \to |\lambda_1|} P(\lambda) = 0$, 
we have by $\Lambda = \pm 2\sqrt{P(\lambda)}$ that 
$$
[-8 |\lambda_1| |\lambda_2|^2,8 |\lambda_1| |\lambda_2|^2] i
$$
also belongs to the stability spectrum. Hence, the image of $\Lambda = \pm 2\sqrt{P(\lambda)}$ covers all $i \mathbb{R}$.	

It remains to prove that the image of $\Lambda = \pm 2\sqrt{P(\lambda)}$ for $\lambda \in \Sigma_+$ in the Lax spectrum (\ref{Lax-cn}) is not contained in $i \mathbb{R}$. Since $\Sigma_-$ is a reflection of $\Sigma_+$ about $\lambda = 0$ and $\Lambda = 0$ if $\lambda = \pm \lambda_2$ or $\lambda = \pm \lambda_3$, the image of $\Lambda = \pm 2\sqrt{P(\lambda)}$ for $\lambda \in \Sigma_-$ is a reflection of the image of $\Lambda = \pm 2\sqrt{P(\lambda)}$ for $\lambda \in \Sigma_+$ about $\Lambda = 0$, so that the spectral bands are not contained in $i \mathbb{R}$ and are symmetric about $\Lambda = 0$.

Let $\lambda \in \Sigma_+$. Due to symmetries between $\{ \pm \lambda_2, \pm \bar{\lambda}_2\}$, we may have only four possibilities, for each we prove that there exists $\lambda \in \Sigma_+$ such that $\Lambda = \pm 2\sqrt{P(\lambda)} \notin i \R$.
\begin{itemize} 
\item $\Sigma_+$ connects $\lambda_2$ and $\bar{\lambda}_2$ and crosses $\R$ outside the segment $[0,|\lambda_1|]$. For $\lambda_0 \in \Sigma_+ \cap \R$, $P(\lambda_0) > 0$ so that $\Lambda = \pm 2\sqrt{P(\lambda_0)} \in \R$.\\

\item $\Sigma_+$ connects $\lambda_2$ and $\bar{\lambda}_2$ and crosses $\R$ inside the segment $(0,|\lambda_1|]$. $\Sigma_+$ and $\R$ intersect perpendicularly at $\lambda_0 \in \Sigma_+ \cap \R$ due to the symmetry of 
Lax spectrum in Proposition \ref{proposition_11} and $P(\lambda_0) \leq 0$. Hence, there is $\lambda \in \Sigma_+$ such that $|\lambda - \lambda_0|$ is small and $|{\rm Re}(\lambda - \lambda_0)| \ll |{\rm Im}(\lambda)|$. Since $P'(\lambda_0) \in \R$, we have ${\rm Re}(P(\lambda)) \leq 0$ and ${\rm Im}(P(\lambda)) \neq 0$ for this $\lambda$ so that $\Lambda = \pm 2\sqrt{P(\lambda)} \notin i\R$.\\

\item $\Sigma_+$ connects $\lambda_2$ and $0$ but does not intersect $(\R \cup i\R) \backslash \{0\}$. Since $P(0) < 0$, $P'(0) = 0$, and $P''(0) \in \R$, we have ${\rm Re}(P(\lambda)) < 0$ and ${\rm Im}(P(\lambda)) \neq 0$ for every $\lambda \in \Sigma_+$ with small $|\lambda|$ and ${\rm Re}(\lambda) {\rm Im}(\lambda) \neq 0$ so that $\Lambda = \pm 2\sqrt{P(\lambda)} \notin i\R$.\\

\item $\Sigma_+$ connects $\lambda_2$ and $-\bar{\lambda}_2$ and crosses $i \R$. $\Sigma_+$ and $i \R$ intersect perpendicularly at $\lambda_0 \in \Sigma_+ \cap i\R$ due to the symmetry of 
Lax spectrum in Proposition \ref{proposition_11} and  $P(\lambda_0) < 0$. Hence, there is $\lambda \in \Sigma_+$ such that $|\lambda - \lambda_0|$ is small and $|{\rm Im}(\lambda - \lambda_0)| \ll |{\rm Re}(\lambda)|$. Since $P'(\lambda_0) \in i \R$, we have ${\rm Re}(P(\lambda)) < 0$ and ${\rm Im}(P(\lambda)) \neq 0$ for this $\lambda$ so that $\Lambda = \pm 2\sqrt{P(\lambda)} \notin i\R$.\\
\end{itemize}

The list of four possibilities above is complete.
\end{proof}

\section{Lax and stability spectra: numerical results}
\label{sec_lax_periodic}

We use the Fourier collocation method, see \cite[Chapter 2, p.45]{Yjk}, in order to approximate Lax spectrum of $\Psi' = L(U,\lambda) \Psi$ numerically 
for the traveling waves with the periodic profile $U$. 
This numerical method has been used in  our previous work \cite{CuiP}.
For every $\lambda$ in the Lax spectrum, the stability spectrum is obtained from $\Lambda = \pm 2\sqrt{P(\lambda)}$ as in (\ref{Lambda}). The numerical results verify the assumptions of Theorems \ref{theorem-stab} and \ref{theorem-instab} and illustrate their conclusions.

\subsection{Waveform (\ref{solution_1})}

We take four real roots $\{ u_1, u_2, u_3, u_4 \}$ of $Q(u) = d$ in the particular setting of $u_1=1$, $u_2=0.5$, $u_3=0$, and $u_4 = -u_1 -u_2-u_3 = -1.5$. The Lax spectrum is shown in Figure \ref{fig_1}(a) in agreement with (\ref{Lax-dn}), where the magenta crosses represent the roots of $P(\lambda)$. The stability spectrum shown in Figure \ref{fig_1}(b) is equivalent to $i \R$ in agreement with Theorem \ref{theorem-stab}.

\begin{figure}[htb!]
	\centering
	\subfigure[Lax spectrum in $\lambda$-plane. ]{\includegraphics[width=2.2in,height=1.8in]{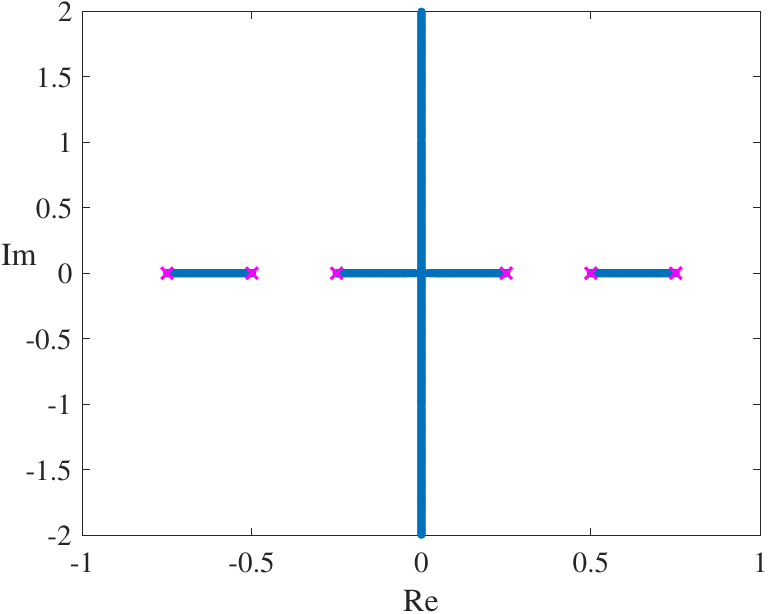}}
	\subfigure[Stability spectrum in  $\Lambda$-plane.]{\includegraphics[width=2.2in,height=1.8in]{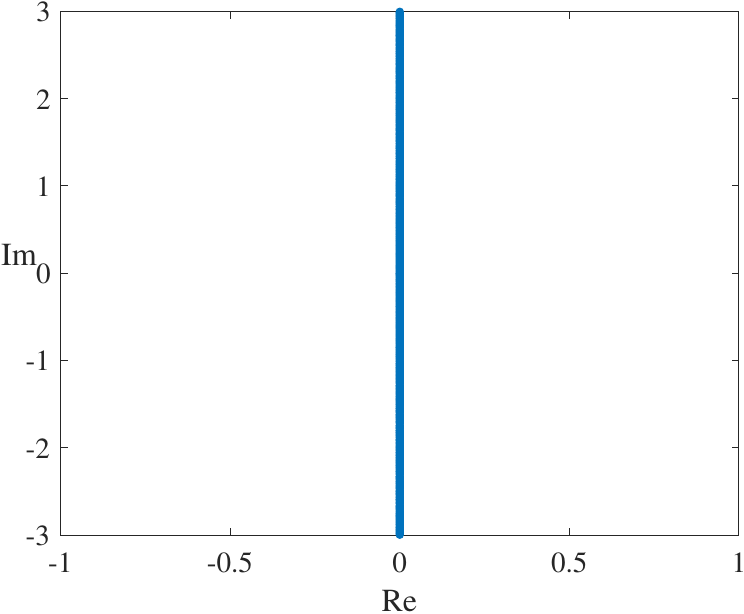}}
	\caption{Numerically computed Lax and stability spectra for the periodic solution with the profile (\ref{solution_1}) for $u_1=1$, $u_2=0.5$, $u_3=0$, and $u_4 = -1.5$.}
	\label{fig_1}
\end{figure}

\subsection{Waveform (\ref{solution_2})}

We take two real roots of $Q(u) = d$ as $u_1=1$ and $u_2=0.2$ and 
the two complex-conjugate roots as $u_3 = \bar{u}_4 = -0.6+0.6i$ 
so that $u_1 + u_2 + u_3 + u_4 = 0$. The Lax spectrum is shown in Figure  \ref{fig_cn2}(a) in agreement with (\ref{Lax-cn}). The stability spectrum is shown in Figure \ref{fig_cn2}(b) in agreement with Theorem \ref{theorem-instab}. The complex bands $\Sigma_{\pm}$ in (\ref{Lax-cn}) are connected across $i\R$ in the Lax spectrum so that $\Sigma_- = \bar{\Sigma}_+$. The stability spectrum is a standard figure-$8$ instability band. 

For a different set of parameter values, the complex bands $\Sigma_{\pm}$ in (\ref{Lax-cn}) are connected across $\R$ in the Lax spectrum so that $\Sigma_- = -\bar{\Sigma}_+$. This is illustrated on Figure  \ref{fig_cn3} for the choice of $u_1=1$, $u_2=-0.2$, and $u_3 = \bar{u}_4 = -0.4+0.2i$. Nevertheless, 
the stability spectrum is still a standard figure-$8$ instability band. 

\begin{figure}[htb!]
	\centering
	\subfigure[Lax spectrum $\lambda$-plane.]{\includegraphics[width=2.2in,height=1.9in]{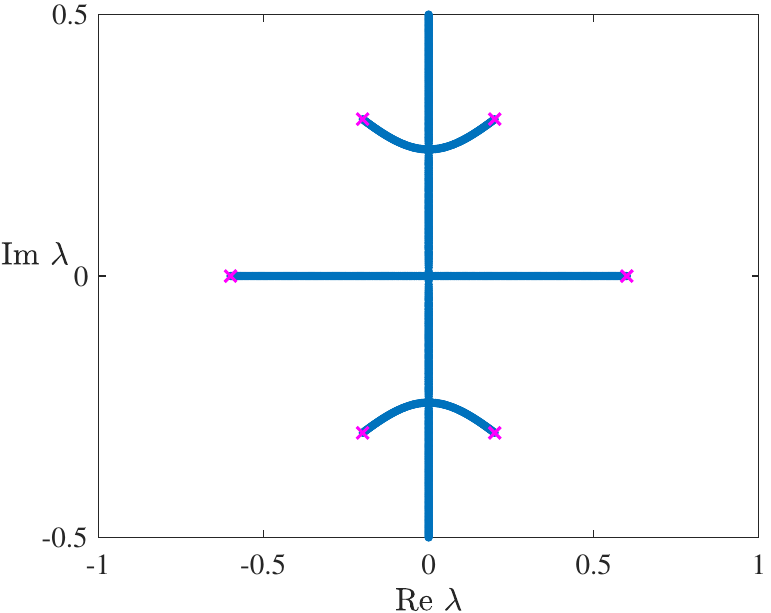}}
	\subfigure[Stability spectrum $\Lambda$-plane.]{\includegraphics[width=2.2in,height=1.9in]{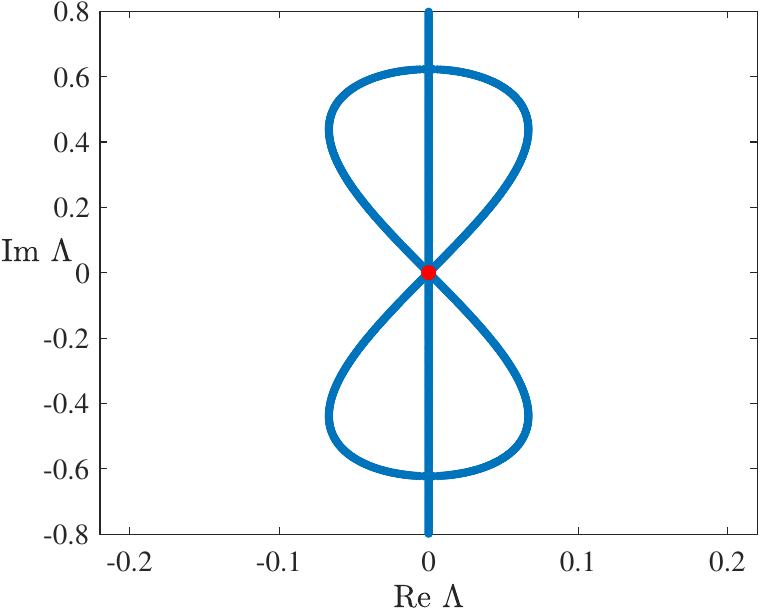}}
	\caption{Numerically computed Lax and stability spectra for the periodic solution with the profile (\ref{solution_2}) for $u_1=1$, $u_2=0.2$, and $u_3 = \bar{u}_4 = -0.6+0.6i$. }
	\label{fig_cn2}
\end{figure}

\begin{figure}[htb!]
	\centering
	\subfigure[Lax spectrum $\lambda$-plane.]{\includegraphics[width=2.2in,height=1.9in]{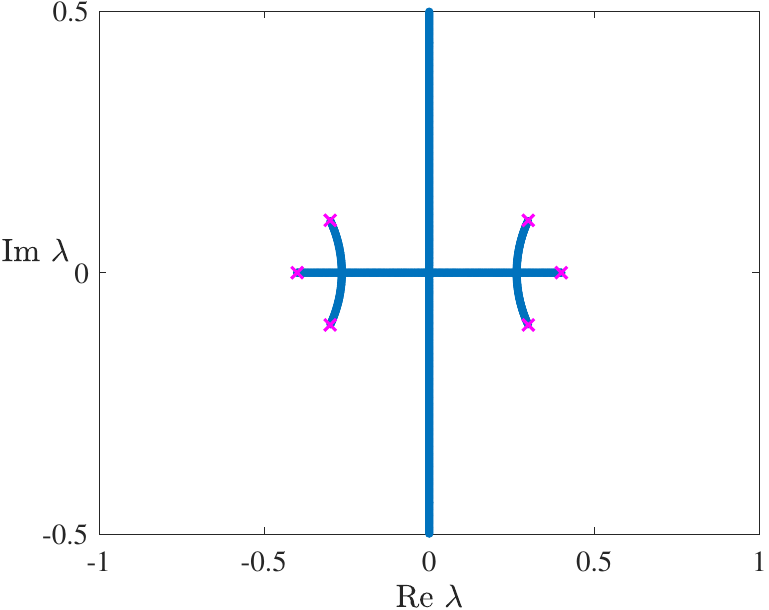}}
	\subfigure[Stability spectrum $\Lambda$-plane.]{\includegraphics[width=2.2in,height=1.9in]{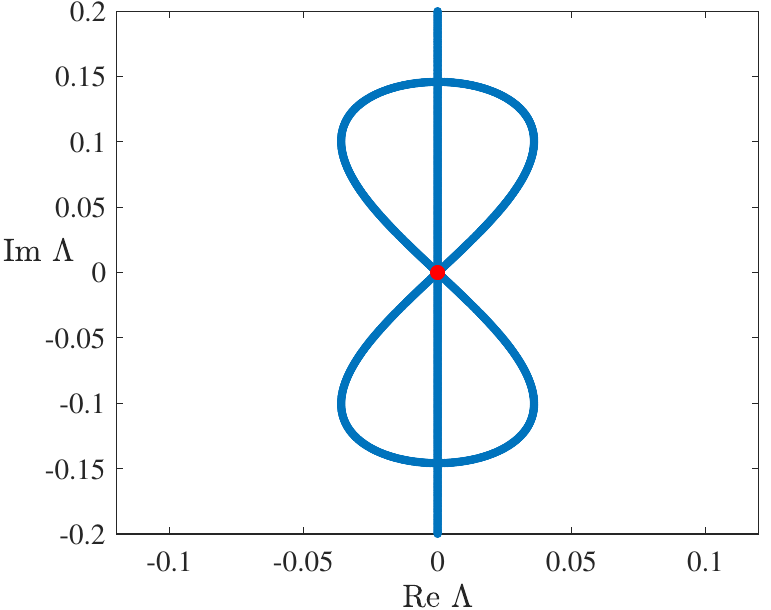}}
	\caption{The same as in Figure \ref{fig_cn2} but for $u_1=1$, $u_2=-0.2$, and 
		$u_3 = \bar{u}_4 = -0.4+0.2i$.}\label{fig_cn3}
\end{figure}

At the first glance, readers may get impression that the cascade of instabilities for the cn-periodic wave (\ref{cn-wave}) shown on Figure \ref{fig_case} is not observed for the periodic waveform (\ref{solution_2}). However, this is just because the two parameter configurations in Figures \ref{fig_cn2} and \ref{fig_cn3} do not represent a general picture.

In order to unfold the cascade of instability bifurcations for the periodic waveform (\ref{solution_2}), we parameterize the roots of $Q(u) = d$ with parameters $\kappa \in (0,1)$ and $\epsilon \in (0,2\kappa)$ as
$$
u_1 = \kappa,\quad  u_2=-\kappa+\epsilon,\quad  \gamma=-\frac{\epsilon}{2}, \quad \eta=\sqrt{1-\kappa^2}.
$$
The periodic waveform (\ref{solution_2}) becomes 
\begin{equation}
\label{solution_2_1}
U(x)=\kappa+\frac{(-2\kappa+\epsilon)\big(1-{\rm cn}(\mu x;k)\big)}{1+\delta+(\delta-1){\rm cn}(\mu x;k)},
\end{equation}
where  
\begin{align*}
\delta&=\frac{\sqrt{\left(\kappa-\frac{3}{2}\epsilon\right)^2+1-\kappa^2}}
{\sqrt{\left(\kappa+\frac{\epsilon}{2}\right)^2+1-\kappa^2}},\\
\mu&=\sqrt[4]{\left[\left(\kappa-\frac{3}{2}\epsilon\right)^2+1-\kappa^2\right]
	\left[\left(\kappa+\frac{\epsilon}{2}\right)^2+1-\kappa^2\right]},\\
2k^2&=1-\frac{\left(\kappa-\frac{\epsilon}{2}\right)\left(-\kappa+\frac{3}{2}\epsilon\right)+1-\kappa^2}{\mu^2}.
\end{align*}
When $\epsilon = 0$, we have $\delta=1$, $\mu=1$, and $k=\kappa$ so that we recover the cnoidal wave with the profile $U(x)=\kappa{\rm cn}(x;\kappa)$ 
as in the periodic solution (\ref{cn-wave}). 

We fix $\kappa=0.97$ and compute numerically the Lax and stability spectra for  different values of $\epsilon \in (0,2\kappa)$. The calculated Lax and stability spectra shown on Figure \ref{fig_4} turns out to be very similar 
to those shown on Figure \ref{fig_case} for the cnoidal wave (\ref{cn-wave}). The only difference from Figure \ref{fig_case} is that figure-$8$ on panel (l) corresponds to the segments $\Sigma_{\pm}$ of the Lax spectrum in Theorem \ref{theorem-instab} crossing the real line on panel (i) and that the line segment $[-|\lambda_1|,|\lambda_1|]$ has a nonzero length. The co-periodic instability (red points on the stability spectrum) arises when the segments $\Sigma_{\pm}$ of the Lax spectrum touch the end points of the line segment $[-|\lambda_1|,|\lambda_1|]$.

\begin{figure}[htb!]
	\centering
	\subfigure[$\epsilon=0.1$]{\includegraphics[width=1.8in,height=1.4in]{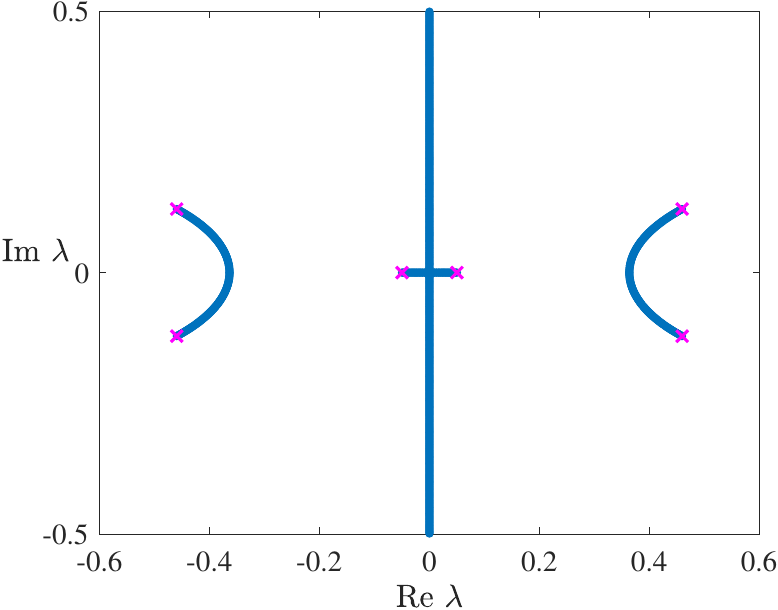}}
	\subfigure[$\epsilon=0.3$]{\includegraphics[width=1.8in,height=1.4in]{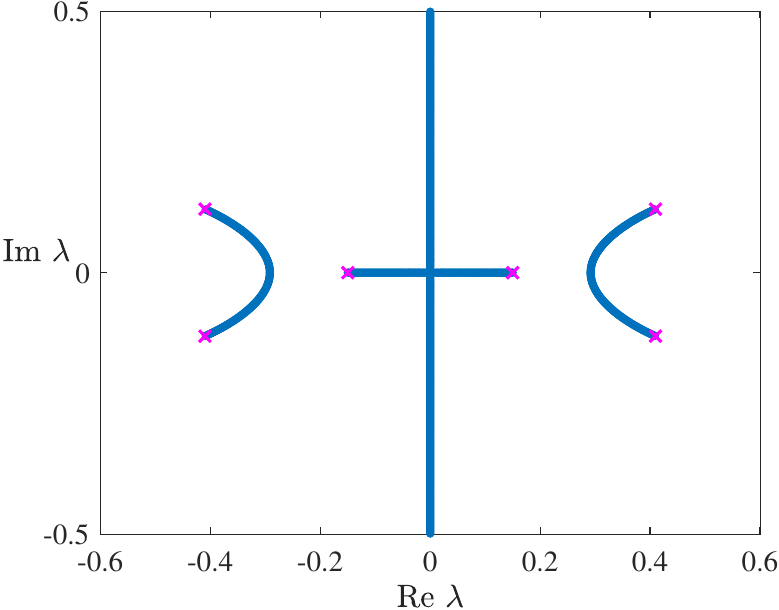}}
		\subfigure[$\epsilon=0.4$]{\includegraphics[width=1.8in,height=1.4in]{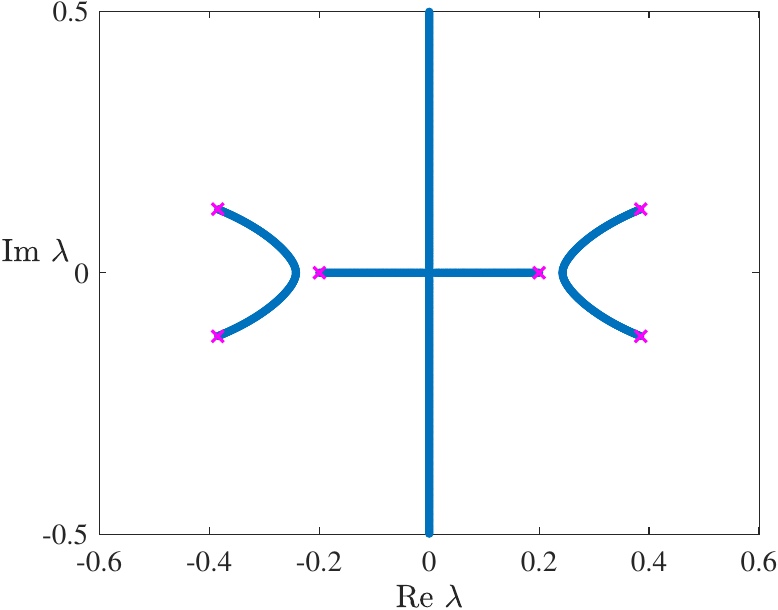}}
	\subfigure[$\epsilon=0.1$]{\includegraphics[width=1.8in,height=1.4in]{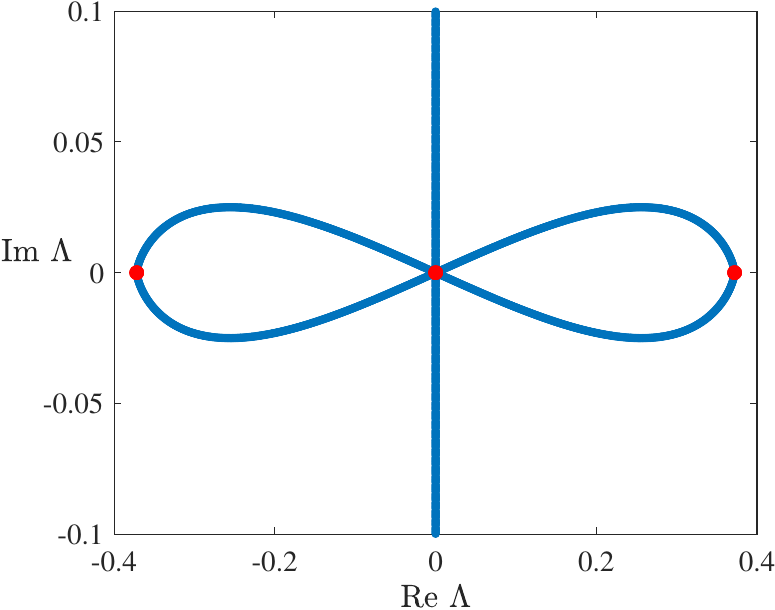}}
	\subfigure[$\epsilon=0.3$]{\includegraphics[width=1.8in,height=1.4in]{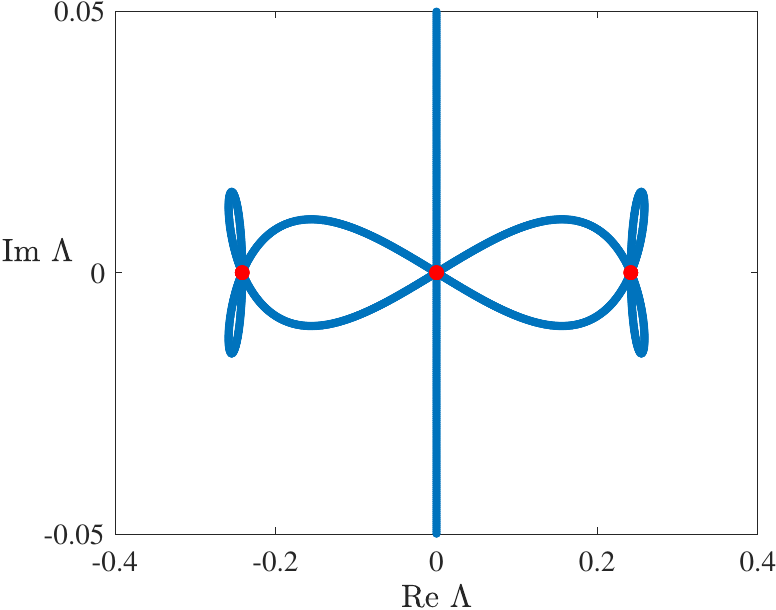}}
		\subfigure[$\epsilon=0.4$]{\includegraphics[width=1.8in,height=1.4in]{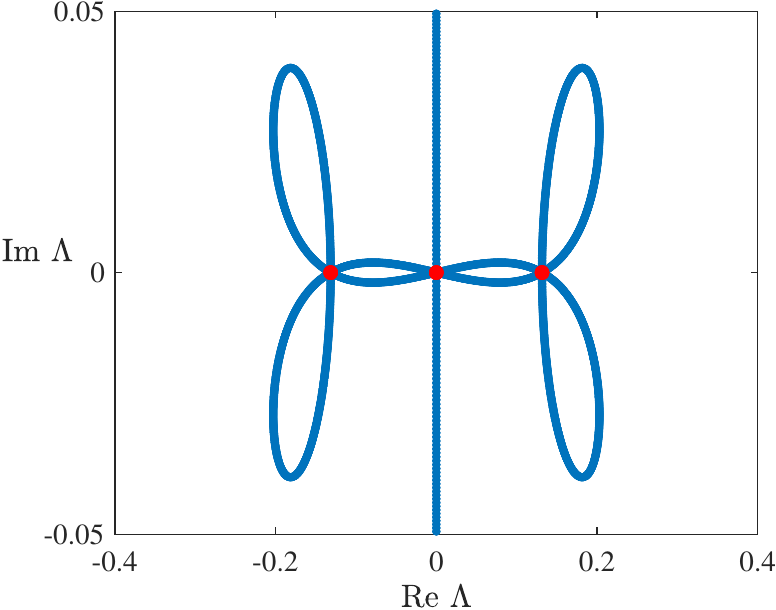}}
	\vspace{-0.2cm}	
		\subfigure[$\epsilon=0.43$]{\includegraphics[width=1.8in,height=1.4in]{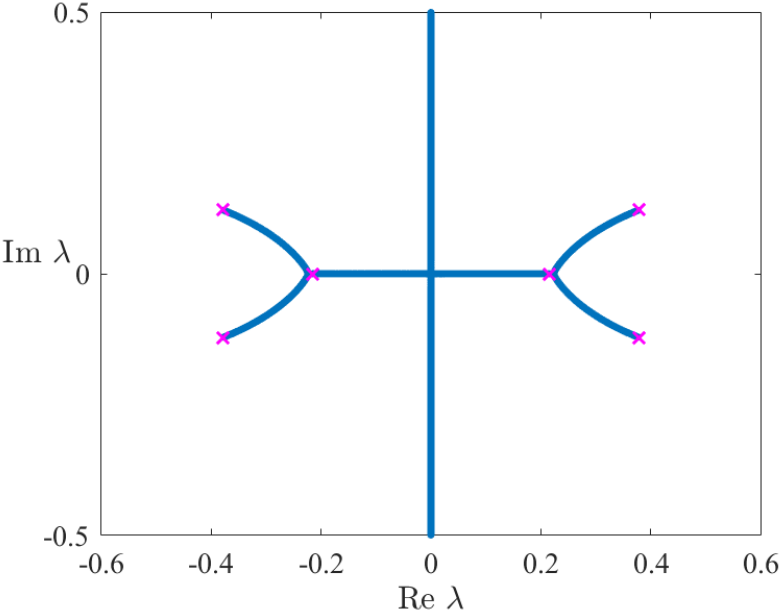}}
			\subfigure[$\epsilon=0.46$]{\includegraphics[width=1.8in,height=1.4in]{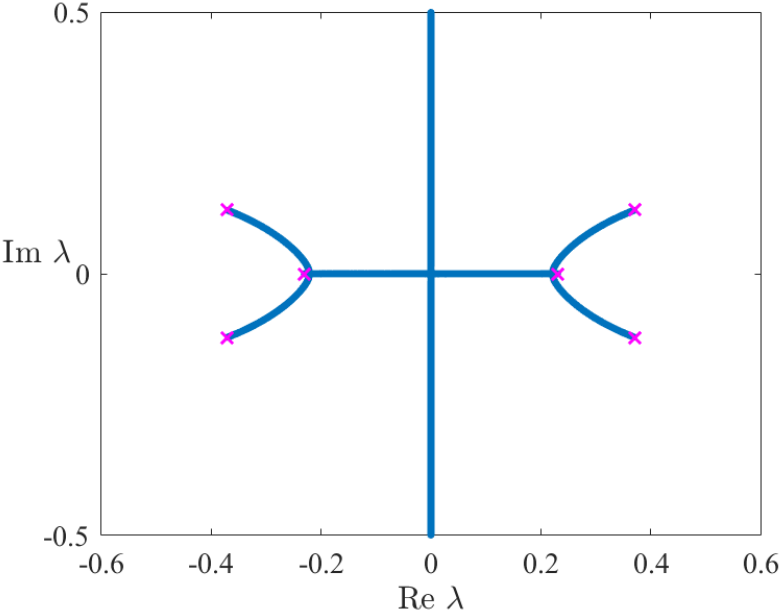}}
			\subfigure[$\epsilon=0.6$]{\includegraphics[width=1.8in,height=1.4in]{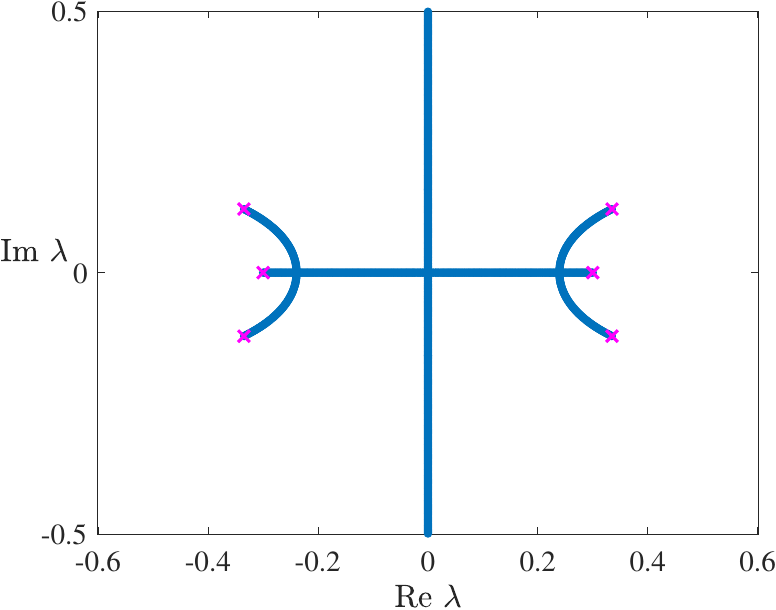}}
			\subfigure[$\epsilon=0.43$]{\includegraphics[width=1.8in,height=1.4in]{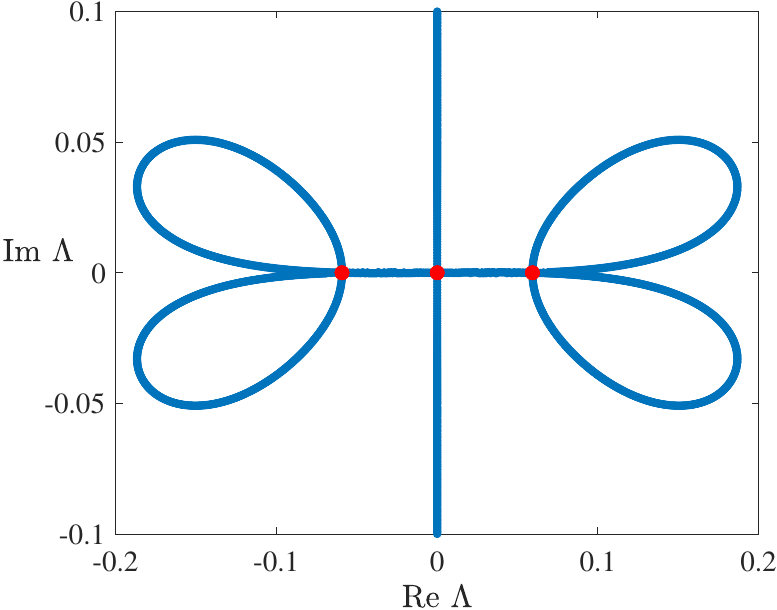}}
				\subfigure[$\epsilon=0.46$]{\includegraphics[width=1.8in,height=1.4in]{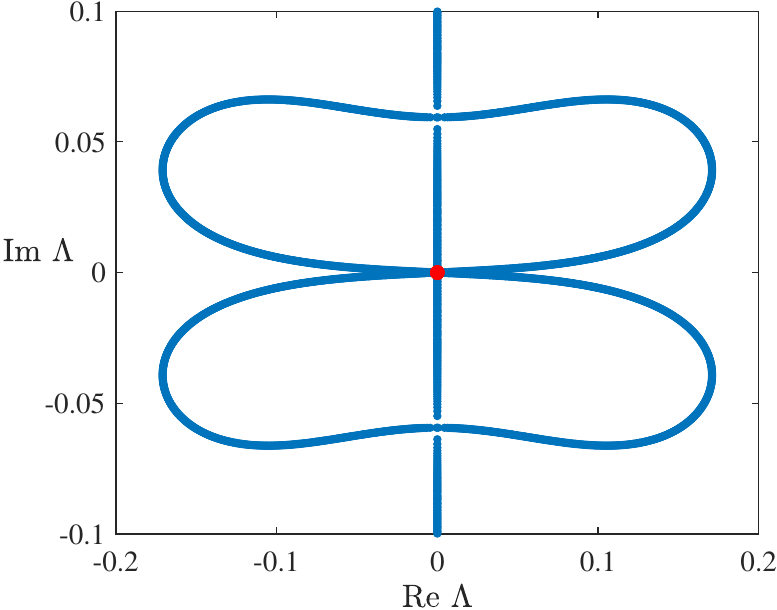}}
			\subfigure[$\epsilon=0.6$]{\includegraphics[width=1.8in,height=1.4in]{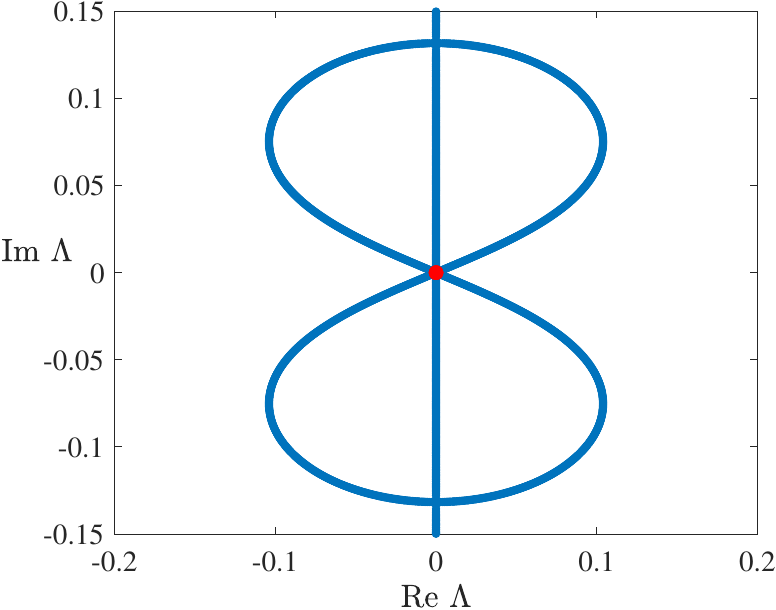}}
	\caption{The Lax and stability spectra for periodic waveform (\ref{solution_2_1}) with $\kappa=0.97$ and different values of $\epsilon$. (a)-(c) and (g)-(i): Lax spectrum in $\lambda$-plane; (d)-(f) and (j)-(l): stability spectrum  in $\Lambda$-plane.}
	\vspace{-0.2cm}
	\label{fig_4}
\end{figure}

If we fix $\kappa=0.9$ and change $\epsilon$ to the negative values as well, 
then the computed Lax and stability spectra shown on Figure \ref{fig_5} features a different transformation of the instability bands. 
The figure-$8$ instability on panel (f) is related to the segments $\Sigma_{\pm}$ of the Lax spectrum in Theorem \ref{theorem-instab} 
crossing the imaginary line on panel (c). The co-periodic instability (red points on the stability spectrum) arises again when the segments $\Sigma_{\pm}$ of the Lax spectrum touch the end points of the line segment $[-|\lambda_1|,|\lambda_1|]$. The co-periodic instability is present when the segments $\Sigma_{\pm}$ intersect the real line outside $[-\lambda_1,\lambda_1]$ and is absent when they intersect the real line inside $[-\lambda_1,\lambda_1]$.

For the cnoidal wave on Figure \ref{fig_case}, the line segment $[-|\lambda_1|,|\lambda_1|]$ shrinks to the origin  and the co-periodic instability arises when the segments $\Sigma_{\pm}$ of the Lax spectrum touch the origin. Hence, all four cases analyzed in the proof of Theorem \ref{theorem-instab} do actually occur in the Lax spectrum for the waveform (\ref{solution_2}) with different parameter values.

\begin{figure}[htb!]
	\centering
	\subfigure[$\epsilon=-0.3$]{\includegraphics[width=1.8in,height=1.4in]{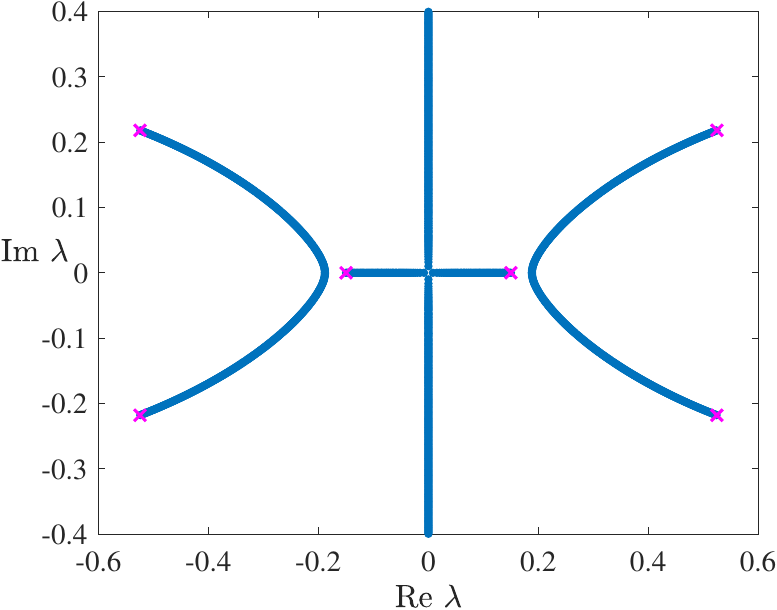}}
	\subfigure[$\epsilon=-0.05$]{\includegraphics[width=1.8in,height=1.4in]{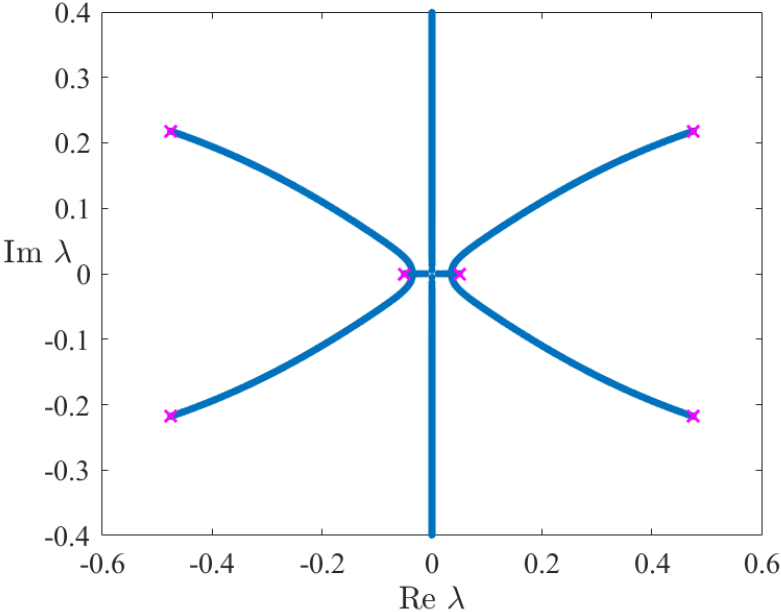}}
	\subfigure[$\epsilon=0.2$]{\includegraphics[width=1.8in,height=1.4in]{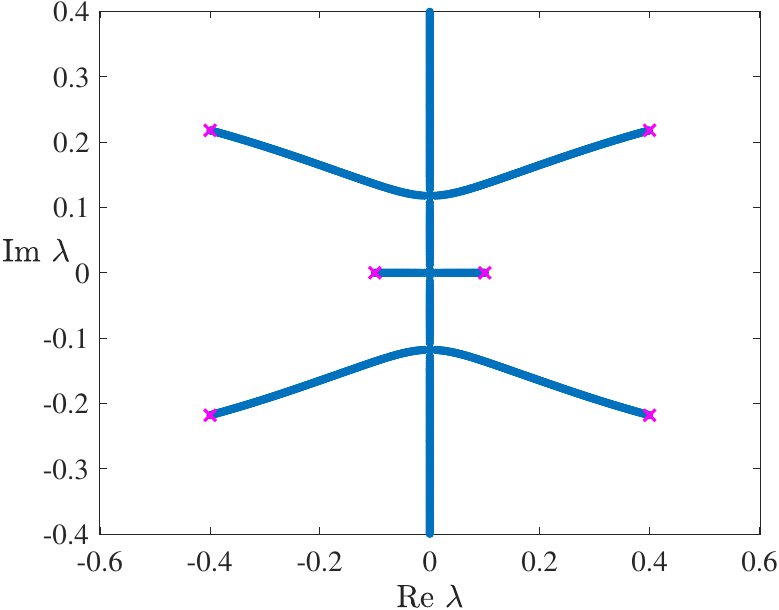}}
	\vspace{-0.2cm}	
	\subfigure[$\epsilon=-0.3$]{\includegraphics[width=1.8in,height=1.4in]{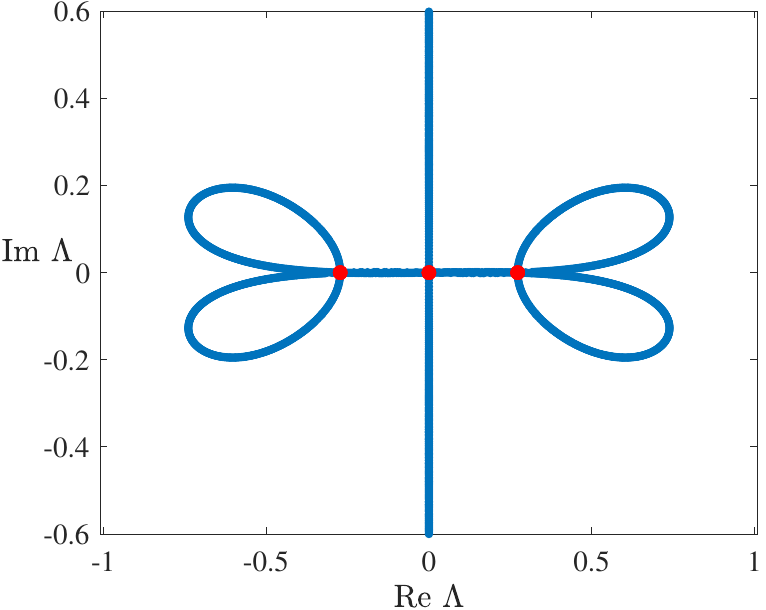}}
	\subfigure[$\epsilon=-0.05$]{\includegraphics[width=1.8in,height=1.4in]{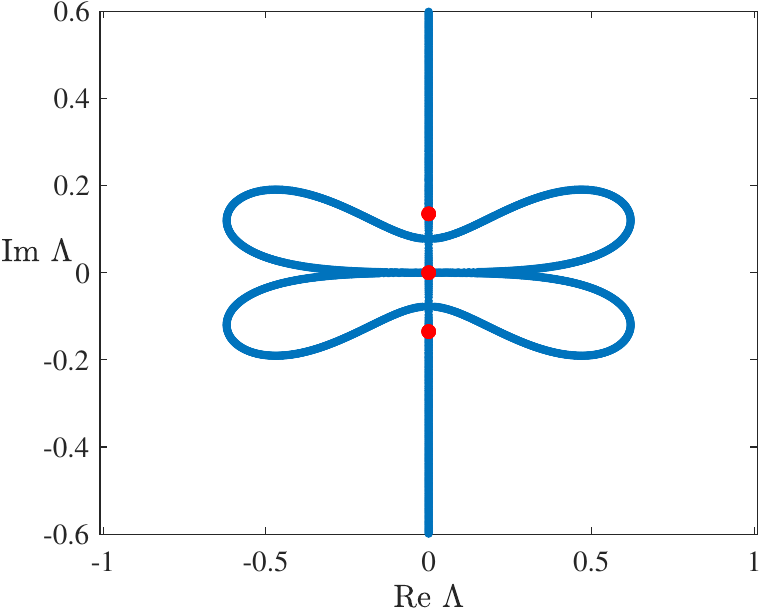}}
	\subfigure[$\epsilon=0.2$]{\includegraphics[width=1.8in,height=1.4in]{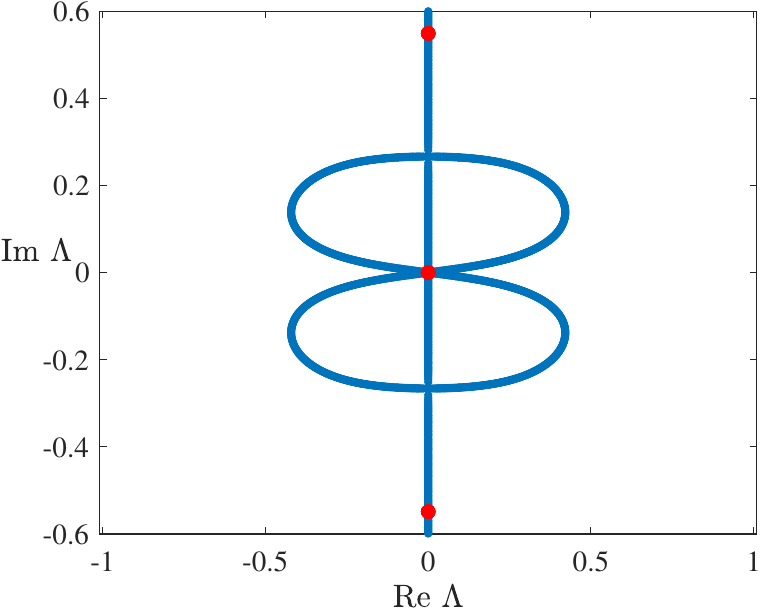}}
	\caption{The Lax and stability spectra for periodic waveform (\ref{solution_2_1}) with $\kappa=0.9$ and different values of $\epsilon$. (a)-(d): Lax spectrum in $\lambda$-plane; (e)-(h): stability spectrum  in $\Lambda$-plane.}
	\vspace{-0.2cm}
	\label{fig_5}
\end{figure}

\section{Conclusion}
\label{sec_modu}

We have studied the spectral stability of the  periodic traveling waves in the focusing mKdV equation and showed that the instability bands for the cnoidal periodic waves transform from figure-$8$ into figure-$\infty$ due to the co-periodic instability bifurcation. This transformation is rather generic for other models with periodic traveling waves (Stokes waves) \cite{Carter,DO11,DDLS24}. 

The conclusion was obtained by using a relation between squared eigenfunctions of the Lax pair and eigenfunctions of the linearized mKdV equation at the periodic traveling waves. The location of the Lax spectrum remains an open problem, especially for the cnoidal periodic waves. It is expected that the elliptic function theory can be useful to compute it explicitly. 

Given universality of the focusing mKdV equation for many applications in fluids, optics, and plasmas, the conclusions obtained in this work can be used for the comprehensive study of the modulational instability of the cnoidal periodic waves in the relevant non-integrable models.

{\bf Acknowledgements.} The work of S. Cui was conducted during PhD studies while visiting McMaster University with the financial support from the China Scholarship Council. The work of D. E. Pelinovsky was supported in part by the National Natural Science Foundation of China (No. 12371248).

\end{document}